\documentclass[journal]{IEEEtran}
\usepackage{enumitem}
\usepackage{multirow}
\usepackage{textcomp}
\usepackage{color}
\usepackage{booktabs}
\usepackage{threeparttable}

\usepackage{cite}
\ifCLASSINFOpdf
  \usepackage[pdftex]{graphicx}
\else
  \usepackage[dvips]{graphicx}
\fi

\usepackage{amsmath, amsthm, amssymb, amsfonts}
\usepackage{mathrsfs}
\usepackage{algorithm}
\usepackage{algpseudocode}
\usepackage{array}

\ifCLASSOPTIONcompsoc
  \usepackage[caption=false,font=normalsize,labelfont=sf,textfont=sf]{subfig}
\else
  \usepackage[caption=false,font=footnotesize]{subfig}
\fi

\newtheorem{theorem}{Theorem}
\newtheorem{definition}{Definition}

\begin{document}
%
\title{UNMAS: Multi-Agent Reinforcement Learning for Unshaped Cooperative Scenarios}
%
%
%

\author{Jiajun~Chai, ~\IEEEmembership{Student~Member,~IEEE,}
        Weifan~Li,~\IEEEmembership{Student~Member,~IEEE,}
        Yuanheng~Zhu,~\IEEEmembership{Senior Member,~IEEE,}
        and~Dongbin~Zhao,~\IEEEmembership{Fellow,~IEEE}
        Zhe~Ma,
        Kewu~Sun,
        Jishiyu~Ding
\thanks{J. Chai, W. Li, Y. Zhu, and D. Zhao are with the State Key Laboratory of Management and Control for Complex Systems, Institute of Automation, Chinese Academy of Sciences, Beijing 100190, China, and are also with the School of Artificial Intelligence, University of Chinese Academy of Sciences, Beijing 100049, China. Z. Ma, K. Sun, and J. Ding are with Xlab, the second academy of CASIC, Beijing 100854, China.}
\thanks{J. Chai and W. Li contribute equally to this paper.}}

\maketitle

\begin{abstract}
Multi-agent reinforcement learning methods such as VDN, QMIX, and QTRAN that adopt centralized training with decentralized execution (CTDE) framework have shown promising results in cooperation and competition. However, in some multi-agent scenarios, the number of agents and the size of action set actually vary over time. We call these \emph{unshaped scenarios}, and the methods mentioned above fail in performing satisfyingly.
In this paper, we propose a new method called Unshaped Networks for Multi-Agent Systems (UNMAS) that adapts to the number and size changes in multi-agent systems. We propose the self-weighting mixing network to factorize the joint action-value. Its adaption to the change in agent number is attributed to the nonlinear mapping from each-agent Q value to the joint action-value with individual weights. Besides, in order to address the change in action set, each agent constructs an individual action-value network that is composed of two streams to evaluate the constant environment-oriented subset and the varying unit-oriented subset. We evaluate UNMAS on various StarCraft II micro-management scenarios and compare the results with several state-of-the-art MARL algorithms. The superiority of UNMAS is demonstrated by its highest winning rates especially on the most difficult scenario 3s5z\_vs\_3s6z. The agents learn to perform effectively cooperative behaviors while other MARL algorithms fail in. Animated demonstrations and source code are provided in https://sites.google.com/view/unmas.
\end{abstract}


\IEEEpeerreviewmaketitle

\section{Introduction}
\IEEEPARstart{M}{ulti-agent} reinforcement learning (MARL) employs reinforcement learning to solve the multi-agent system problems. With cooperative multi-agent systems playing an increasingly important role, such as controlling robot swarms with limited sensing capabilities \cite{Cao2013a, Jiang2019}, mastering build order production \cite{Tang2018b} and micro-management task in real-time strategy (RTS) games \cite{Shao2019, Rashid2018, Tang2019}, dispatching ride requests \cite{Li}, autonomous driving \cite{Zhu2018, Liang2019, Li2019a}, and so on, MARL attracts the attention of many researchers. MARL algorithms face the problem of huge action space, non-stationary environment, and global exploration. Although many algorithms have been proposed to solve these problems, it is still an open problem.

To address cooperative tasks, one class of MARL methods is independent learning that allows the agents to learn independently \cite{Palmer2018, Tampuu2017, Gupta2017}. Independent learning suffers from the non-stationarity because other agents are also impacting the environment. Another class is centralized learning that takes the multi-agent system as a whole \cite{Foerster2018, Shao2018a, Zhang2016}. \emph{Centralized training with decentralized execution} (CTDE) is a compromise between independent and centralized learning \cite{Zhang2018, Sunehag2018, Rashid2018}. It provides local autonomy to agents by decentralized execution and mitigates the problem of the non-stationary environment by centralized training. There are many methods adopting CTDE framework \cite{Wang2020, Yang2020, Wang2020a, Yang2020a}, including QMIX \cite{Rashid2018}, VDN \cite{Sunehag2018}, QTRAN \cite{Son2019}, and so on.

\emph{Unshaped scenario} is defined as the scenario where the number of units and the size of action set change over time, which is a common multi-agent scenario.  \cite{Sui2020} solves the problem of formation control in face of an uncertain number of obstacles. DyAN \cite{wang2020few} reconstructs the agent observation into information related to environment and units, and ASN \cite{wang2020action} categorizes the agent action set considering the semantics between actions. Although the mentioned methods adapt to unshaped scenarios, there is still room for improvement. 

StarCraft II micro-management is a widely used experimental environment to test MARL methods \cite{Tang2019}. It is actually an unshaped scenario, in which the death of agents in combat leads to the change in their number. In addition, since the enemy can also be killed, the size of attack action subset of each unit changes as well. However, the existing methods \cite{Rashid2018, Son2019} ignore the variation of action set and still provide action-values for invalid actions. This may cause miscalculation of the action-value among all agents. In addition, as their joint action-value functions still collect the meaningless action-values of dead agents, the joint action-value may also be miscalculated.

\subsection{Contribution}
In this paper, we focus on the unshaped cooperative multi-agent scenarios, in which the number of agents and the size of action set change over time. Our contribution to the uncertain challenge of MARL is twofold. (i) The self-weighting mixing network is proposed with the network input dimension adapting with the number of agents, so that the joint action-value is estimated more accurately. (ii) The action set of an agent is divided into two subsets: \emph{environment-oriented} subset and \emph{unit-oriented} subset. We use two network streams, the \emph{environment-oriented stream} and \emph{unit-oriented stream}, to calculate the Q values of corresponding actions. Finally, we introduce the training algorithm of UNMAS and conduct experiments to compare with several other MARL algorithms in the StarCraft II micro-management scenarios. The results show that UNMAS achieves the highest winning rates on a wide variety of StarCraft II micro-management maps.

\subsection{Related Work}
Research on independent learning begins with the work of \cite{Tan1993}, in which the author tries to execute Q-learning independently for each agent to achieve cooperation. With the use of deep learning, IQL \cite{Tampuu2017} introduces Deep Q-Network (DQN) \cite{Mnih2015} into MARL to cope with high-dimensional observations. Some other research that tackles the multi-agent problem with independent learning can be found in \cite{Gupta2017, Palmer2018, Zhou2019, zhu2020online}. DyAN \cite{wang2020few} divides the observation of the agents into environment-oriented and unit-oriented subsets, and employs a GNN to adapt to the change in the number of agents in the expansion from small-scale scenarios to large-scale scenarios. However, independent learning suffers from the non-stationarity of environment, which leads to the difficulty of learning convergence \cite{Yang2020}.

Centralized learning treats the multi-agent system as a whole. Grid-Net \cite{Han2019} employs an encoder-decoder architecture to generate actions for each grid agent. \cite{Peng2017} and \cite{Qin2018} propose methods for the multi-agent system in the framework of actor-critic. \cite{Sun2020} formulates the problem of multitask into a multi-agent system. However, the centralized learning method is hard to be scaled to larger scenarios, so the CTDE framework becomes popular as a compromise between independent and centralized learning. One spectrum of the CTDE method is the actor-critic method. MADDPG \cite{NIPS2017_7217}, which is developed from DDPG \cite{Lillicrap2016}, provides a centralized critic for the whole system and a group of decentralized actors for each agent. COMA \cite{Foerster2018} computes a counterfactual baseline to marginalize out a single agent's action while the other agents' actions are fixed. COMA also provides a centralized critic to estimate the joint Q-function.

The other spectrum is the value-based method. The challenge for value-based methods in the framework of CTDE is how to factorize the joint action-value correctly \cite{Yang2020}. VDN \cite{Sunehag2018} factorizes the joint action-value by summing up the individual action-values of each agent. QMIX \cite{Rashid2018} combines the individual action-values in a non-linear way by a mixing network whose weights and biases are generated according to the multi-agent global state. ASN \cite{wang2020action} improves the individual action-value network in VDN and QMIX by considering action semantics between agents. QTRAN \cite{Son2019} guarantees more general factorization by factorizing the joint action-value with a transformation. Q-DPP \cite{Yang2020} introduces DPP into MARL tasks to increase the diversity of agents' behaviors. ROMA \cite{Wang2020} and RODE \cite{Wang2021rode} allow agents with a similar role to share similar behaviors. Qatten \cite{Yang2020b} designs a multi-head attention network to approximate joint action-value function. However, such methods do not consider the problem of unshaped scenarios. They assume that the number of agents and the size of action set are fixed, and this assumption limits the applications of these methods.

\subsection{Organization}
This paper is organized as follows. Section II introduces the background knowledge about multi-agent reinforcement learning and factorization of the joint action-value. Section III describes UNMAS from three aspects: joint action-value function, individual action-value function, and training algorithm. Section IV shows the experiments and results, and analyzes the learned strategies. Finally, Section V gives the conclusion.

\section{Background}
\subsection{Multi-Agent Reinforcement Learning}
We consider a fully cooperative multi-agent task with partially observable environment, in which agents observe and take actions individually. This task is also called the \emph{decentralized partially observable Markov decision process} (Dec-POMDP) \cite{Stroock2015}. It can be defined as a tuple $\mathscr{U} = \{\mathbb{D}, \mathbb{S}, \mathbb{U}, \mathbb{T}, \mathbb{O}, R, \gamma\}$. $\mathbb{D}=\{1,...,n\}$ is the set of agents, the number of which is \emph{n}. The Dec-POMDP extends POMDP by introducing the set of joint actions $\mathbb{U}$ and joint observations $\mathbb{O}$. The multi-agent system takes the joint action $\textbf{u}_t = \{u_{1, t},...,u_{n, t}\}$ according to the joint observation $\textbf{o}_{t} = \{o_{1, t},...,o_{n, t}\}$ and gets the immediate reward $r_t$ from environment according to the function $R: \mathbb{S} \times \mathbb{U} \rightarrow \mathbb{R}$. Then, the global state of multi-agent system $s_{t} \subseteq \mathbb{S}$ is produced according to the transition function $\mathbb{T}$, which specifies Pr$(s_{t+1}|s_t, \textbf{u}_t)$. Finally, $\gamma$ in tuple $\mathscr{U}$ denotes the discount factor of \emph{discounted cumulative reward}: $G_{t} = \sum_{j=0}^{\infty}\gamma^{j}r_{t+j}$. 

In the Dec-POMDP, we consider a joint policy $\boldsymbol{\pi}$, which is composed of the policies $\pi_{i}(u_{i, t}| o_{i, t})$ of every agent $i \in \mathbb{D}$. The joint policy has a joint action-value function: $\textbf{Q}(\textbf{o}_{t}, \textbf{u}_{t}) = \mathbb{E}_{\textbf{o}_{t+1:\infty}, \textbf{u}_{t+1:\infty}}[G_{t}|\textbf{o}_{t}, \textbf{u}_{t}]$. The purpose of fully cooperative multi-agent task is to maximize this return. 

\subsection{Factorization of Joint Action-Value Function}
As mentioned before, value-based methods with CTDE framework need to find an efficient and adaptable way to factorize the joint action-value. A common requirement in the field of CTDE is the \emph{Individual Global Max} (IGM) condition. 

\begin{definition}
For a multi-agent system with $n$ agents, if the optimal joint action is equivalent to the set of agents' actions that make the individual action-value functions get the maximum values, the system is said to satisfy the IGM condition. This statement is formulated as follows:
\begin{equation}
\normalfont
\arg\max_{\textbf{u}_{t}}\textbf{Q}(\textbf{o}_{t}, \textbf{u}_{t}) =
\begin{pmatrix}
\arg\max_{{u}_{1, t}}Q_{1}({o}_{1, t}, {u}_{1, t}) \\
\vdots \\
\arg\max_{{u}_{n, t}}Q_{n}({o}_{n, t}, {u}_{n, t}) \\
\end{pmatrix}.
\end{equation}
\end{definition}

Since the IGM condition is difficult to be verified in practice, the following monotonicity condition is mostly used as its substitute: 
\begin{equation}
\frac{\normalfont \partial \textbf{Q}(\textbf{o}_{t}, \textbf{u}_{t})}{\partial Q_i(o_{i, t}, u_{i, t})} \ge 0, \  \forall i \in \mathbb{D}.
\label{eq:mono}
\end{equation}

If $\normalfont \textbf{Q}(\textbf{o}_{t}, \textbf{u}_{t})$ is factorized monotonically as (\ref{eq:mono}), then this way of factorization meets the IGM condition. Under the IGM condition, maximizing the joint action-value is equivalent to maximizing the action-value of each agent. Thus, the multi-agent system pursues the same goal as every agent and achieves cooperation. 

We propose UNMAS in the next section to adapt to the number and size changes in the unshaped scenario.

\section{Unshaped Networks for Multi-Agent Systems}
In this section, we propose a new method called Unshaped Networks for Multi-Agent Systems (UNMAS), aiming at helping agents adapt to the change in the unshaped scenario. UNMAS uses the self-weighting mixing network, which is adaptive to the size of input, to factorize the joint action-value. The individual action-value network of agent is specially designed with two network streams to evaluate the actions in the \emph{environment-oriented} subset and \emph{unit-oriented} subset separately. The size of unit-oriented subset is also unshaped. 

\subsection{Self-weighting mixing network}
Fig. \ref{fig:self-weighting} presents the self-weighting mixing network, which approximates the joint action-value function. Since UNMAS adopts the CTDE framework, the joint action-value function is only used in training. The weights and biases of the self-weighting mixing network are produced by a group of hyper networks represented by the yellow blocks in the architecture, and the input of the hyper networks is the global state $s_t$. 

\begin{figure}[htbp]
\centering
\includegraphics[scale=0.275]{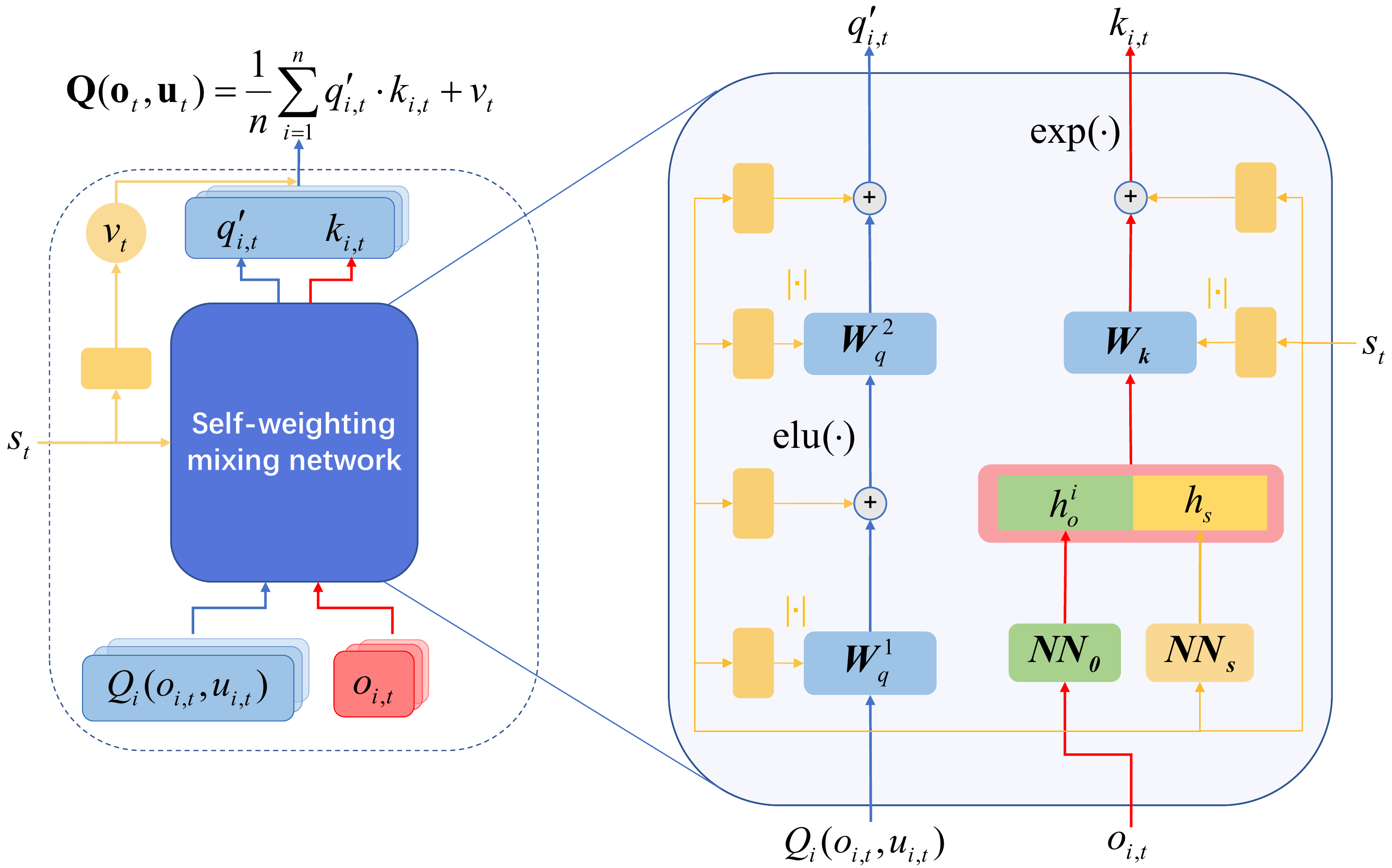}
\caption{Self-weighting mixing network architecture. Specifically, the right side of the diagram shows the details of the self-weighting mixing network. Given the observation and action-value of each agent $i$, it provides $q_{i, t}^{'}$ and $k_{i, t}$ to calculate the joint action-value $\textbf{Q}(\textbf{o}_{t}, \textbf{u}_{t})$ with its bias $v_t$ on the left side. }
\label{fig:self-weighting}
\end{figure}

Instead of following the original definition of the joint action-value function that takes joint observations and joint actions as input, CTDE decomposes it as a mapping from the individual action-values of each agent. Such way greatly decouples the complicated interplay between different agents.

The biggest difference between UNMAS and other CTDE methods lies in the ways of dealing with the input size of joint action-value function. In the other CTDE methods, the input size of joint action-value function is determined beforehand and keeps constant during the whole training and execution phase. Since they still provide action-values for dead agents, the joint action-value may be miscalculated and reduce the performance of agent policies. However, the joint action-value function represented by self-weighting mixing network can change its input size with the number of agents. It outputs two scalars $q_{i, t}^{'}$ and $k_{i, t}$ for each agent:
\begin{equation}
\begin{aligned}
q_{i, t}^{'} &= W_q^2 \cdot {\rm elu} (W_q^1 \cdot Q_i(o_{i, t}, u_{i, t}; \theta_i) + b_q^1) + b_q^2 \\
k_{i, t} &= W_k \cdot [h_{o}^{i}, h_{s}] + b_k,
\end{aligned}
\label{eq:q_and_k}
\end{equation}
where $q_{i, t}^\prime$ is the result after nonlinear mapping by self-weighting mixing network, and $k_{i, t}$ is the individual weight, which describes the contribution of agent $i$ to the joint action-value. By dynamically evaluating the contribution of each agent according to $k_{i, t}$, the estimation of the joint action-value could be more accurate. Thus, the cooperation of the multi-agent system becomes better. Elu \cite{Clevert2015} is a nonlinear activation function, which is also used in the mixing network of QMIX. $Q_i(o_{i, t}, u_{i, t}; \theta_i)$ is the action-value of agent $i$, which is estimated by the \emph{individual action-value network}. $W_q^1$, $W_q^2$, and $W_k$ are weights of the self-weighting mixing network as shown in Fig. \ref{fig:self-weighting}, and $b_q^1$, $b_q^2$, and $b_k$ are the corresponding biases. $h_o^i$ and $h_s$ are outputs of networks $\textbf{NN}_o$ and $\textbf{NN}_s$, which take observation $o_{i, t}$ of agent $i$ and the global state $s_t$ as input, respectively. Finally, the concatenation of $h_o^i$ and $h_s$ is taken as the input of $[W_k, b_k]$ to calculate the weight $k_{i ,t}$.

By employing self-weighting mixing network to evaluate the contribution of each agent, UNMAS can adapt to the change in the number of agents in training and improve the accuracy of the estimation of joint action-value. Taking the StarCraft II micro-management scenario as an example, the number of agents decreases due to the death caused by enemies attack. The weights and biases of the self-weighting mixing network are obtained through the global state. Therefore, even if the $q_{i, t}^{'}$ of each agent do not take the information of other agents into account, the self-weighting mixing network can still estimate the joint action-value accurately. The joint action-value function is formulated as follows:
\begin{equation}
\begin{aligned}
\textbf{Q}(\textbf{o}_{t}, \textbf{u}_{t}; \boldsymbol{\theta}_{joint}) = \frac{1}{n} \sum_{i=1}^{n} q_{i, t}^{'} \cdot k_{i, t} + v_{t},
\end{aligned}
\label{eq:q_jt}
\end{equation}
where $n$ is the number of agents in the current timestep $t$, and $v_{t}$ is a scalar output from the hyper network. It is a bias of the joint action-value, which is generated from a fully-connected network with $s_t$ as input. The ablation result of QMIX \cite{Rashid2018} shows that adding $v_{t}$ to the joint action-value is able to achieve better results, which is also confirmed in our ablation experiments. In (\ref{eq:q_jt}), in order to adapt to the change in the number of agents, we use $n$ to average the weighted sum $\sum_{i=1}^{n} q_{i, t}^{'} \cdot k_{i, t}$. Besides, the weights of self-weighting mixing network take the \emph{absolute} values to make sure that the factorization meets the IGM condition. The proof that the self-weighting mixing network meets the IGM condition is provided in Theorem \ref{theorem}.

\begin{theorem}
In a fully cooperative task, if letting the self-weighting mixing network shown in Fig. \ref{fig:self-weighting} represent the joint action-value function, then the factorization process meets the IGM condition.
\label{theorem}
\end{theorem}

\begin{proof}
The joint action-value estimated by the self-weighting mixing network is written as follows:
\begin{equation}
\begin{aligned}
\textbf{Q}(\textbf{o}_{t}, \textbf{u}_{t}; \boldsymbol{\theta}_{joint}) &= \frac{1}{n} \sum_{i=1}^{n} q_{i, t}^{'} \cdot k_{i, t} + v_{t}\\
&= \frac{1}{n} \sum_{i=1}^{n} [W_q^2 \cdot {\rm elu} (W_q^1 \cdot Q_i + b_q^1) + b_q^2] \cdot \\
& \ \ \ \ \ \ \ \ \ \ \ \ \exp(W_k \cdot [h_{o}^{i}, h_{s}] + b_k) + v_{t}.
\end{aligned}
\label{eq:proof}
\end{equation}
The partial derivative of the joint action-value to agent action-value is shown in (\ref{eq:deri}), where all elements are non-negative.
\begin{equation}
\begin{aligned}
\frac{\partial \textbf{Q}}{\partial Q_i} &= \frac{2}{n} \cdot W_q^2 \cdot \frac{\partial {\rm elu}(W_q^1 \cdot Q_i + b_q^1)}{\partial Q_i} \cdot \\
& \ \ \ \ \ \ \ \ \ \ \ \ \exp(W_k \cdot [h_{o}^{i}, h_{s}] + b_k), \forall i \in \mathbb{D}.
\end{aligned}
\label{eq:deri}
\end{equation}
Since the parameter $\alpha$ of elu function is set to 1, the derivative $\partial {\rm elu}(W_q^1 \cdot Q_i + b_q^1) / \partial Q_i$ shown in (\ref{eq:deri}) can be written as follows:
\begin{equation}
\frac{\partial {\rm elu}(W_q^1 \cdot Q_i + b_q^1)}{\partial Q_i} =
\left\{
\begin{aligned}
W_q^1 \cdot e^{W_q^1 \cdot Q_i + b_q^1} &, W_q^1 \cdot Q_i + b_q^1 \ge 0 & \\
W_q^1\ \ \ \ \ \ \ \ \ \ \ \ \ \ &, otherwise &
\end{aligned}
\right.
\label{eq:deri-elu}
\end{equation}
As described above, the weights of self-weighting mixing network take absolute values, so the partial derivative $\partial \textbf{Q} / \partial Q_i$ is also non-negative. According to the monotonicity condition, the factorization meets the IGM condition.
\label{proof}
\end{proof}

Through the special design of network structure and averaging weighted sum with $n$, self-weighting mixing network can adapt to the change in the number of agents. In the next subsection, we introduce the architecture of the individual action-value network, which can adapt to the change in the size of action set.

\subsection{Individual action-value network}
The individual action-value network approximates the action-value function of each agent, and it shares the parameters among agents. In the framework of CTDE, an agent takes an action according to its local observation without communicating with the other agents. 

In order to help agents adapt to the change in the size of action set, we divide the action set into two subsets. One is the environment-oriented subset, and the other is the unit-oriented subset. The first action subset represents the interactions between the agent and the environment, and keeps constant during the whole running phase. The second represents the interactions between the agent and other agents, and its size varies with the number of agents in the current environment. Taking StarCraft II micro-management scenario as an example, the environment-oriented actions include \emph{stop} and four \emph{movement} actions in the directions of \emph{up}, \emph{down}, \emph{left}, and \emph{right}, while the unit-oriented action is the attack action aiming at an  enemy unit.

\begin{figure}[htbp]
\centering
\includegraphics[scale=0.275]{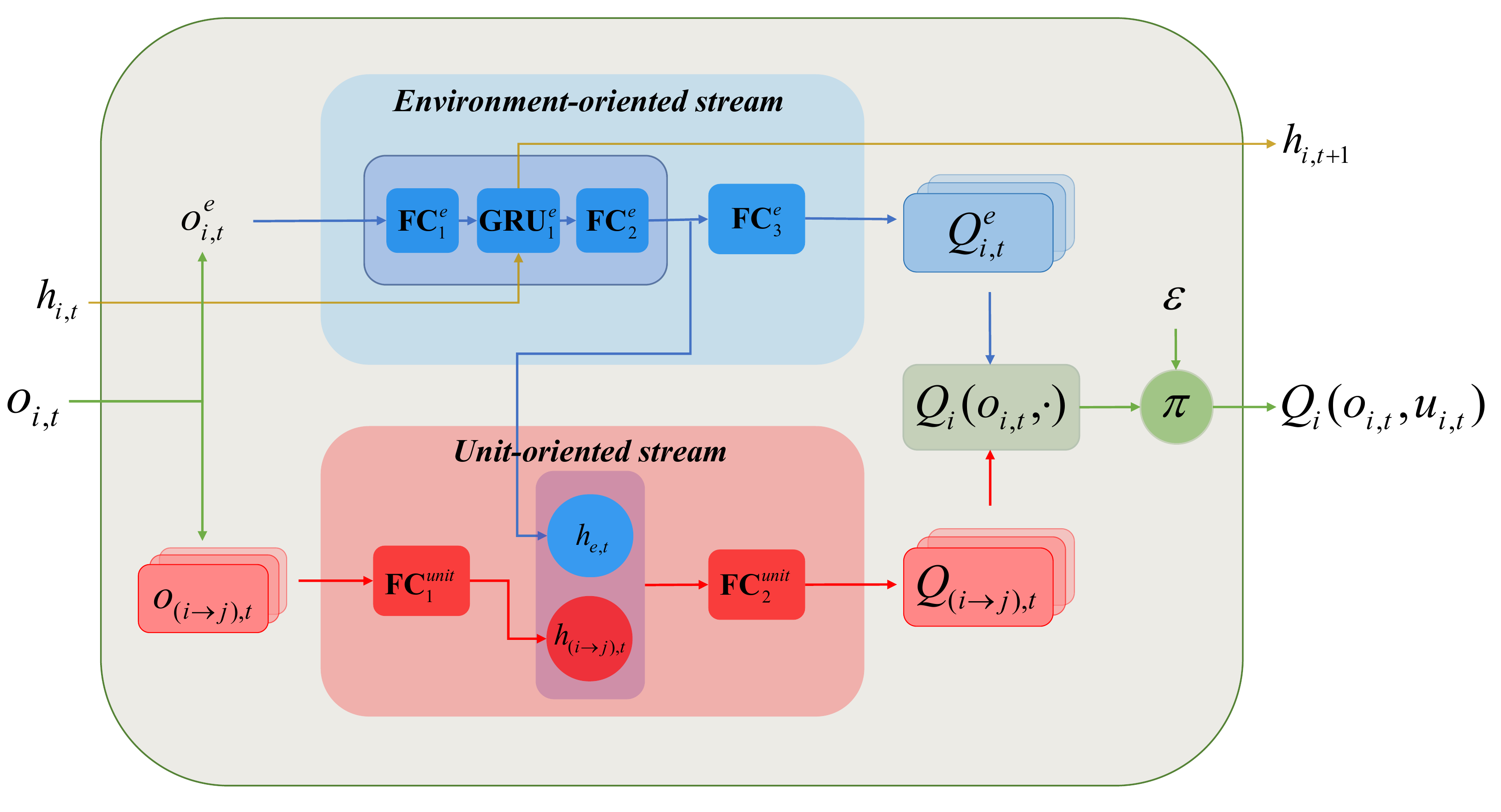}
\caption{Individual action-value network architecture. Specifically, it divides the action set of an agent into two subsets: environment-oriented and unit-oriented, and evaluates them through two separate streams.}
\label{fig:agent}
\end{figure}

Based on the division of the action set, we propose the individual action-value network as shown in Fig. \ref{fig:agent}. Two network streams are constructed for the two action subsets. The \emph{environment-oriented stream} takes the \emph{environment-oriented observation} as input. It contains three \emph{Fully-Connected} (FC) layers $\textbf{FC}_{j}^{e}, j=1,2,3$, and one \emph{Gated Recurrent Unit} (GRU) \cite{Cho2014} layer $\textbf{GRU}_{1}^{e}$.

The GRU layer, employed to solve the problem of POMDP, can better estimate the current state. At each timestep $t$, the GRU layer additionally inputs a hidden state $h_{i, t}$, which contains historical information and outputs the next hidden state. Its outputs are the Q values of the environment-oriented actions as follows:
\begin{equation}
Q_{i, t}^{e}(o_{i, t}^{e}, h_{i, t}, \cdot), h_{i, t + 1} = \textbf{NN}_{i}^{e}(o_{i,t}^{e}, h_{i, t})
\end{equation}
where $Q_{i, t}^{e}$ are the Q values of agent $i$ executing the environment-oriented actions. $o_{i, t}^{e}$ is the environment-oriented observation, which describes the information observed by agent $i$ from the environment. $h_{i, t+1}$ is the hidden state of the next timestep $t+1$. $\textbf{NN}_{i}^{e}$ is the short form of the environment-oriented stream.

The \emph{unit-oriented stream} takes the \emph{unit-oriented observations} as input. Its outputs are the Q values of the unit-oriented actions. The unit-oriented observation $o_{(i\to j), t}$ differs from the environment-oriented observation in the fact that it describes the information from agent $i$ to the target unit $j$. Taking the StarCraft II micro-management scenario as an example, if there are 3 enemies on the map at timestep \emph{t}, the unit-oriented stream takes 3 observations in terms of these enemies to evaluate the Q values of attacking them. Due to the number of enemies changes over time in unshaped scenarios, the number of attack targets also changes. Besides, it should be noted that if the agent has a healing action, the target units here will be other agents, not enemies. Furthermore, we concatenate the first layer's output of two streams to form a new vector, and feed it into the second layer of the unit-oriented stream. The Q values of the unit-oriented actions are calculated as follows:
\begin{equation}
\begin{aligned}
vector &= [h_{e, t}, h_{(i\to j), t}] \\
Q_{(i\to j), t}(o_{(i\to j), t}, u_{(i\to j), t}) &= \textbf{FC}_{2}^{unit}(vector)
\end{aligned}
\label{eq:concat}
\end{equation}
where $vector$ is the concatenation of $h_{e, t}$ and $h_{(i\to j), t}$, which are the output of the first part of environment-oriented stream and unit-oriented stream. $Q_{(i\to j), t}(o_{(i\to j), t}, u_{(i\to j), t})$ is the Q value of agent $i$ executing the unit-oriented action on target unit $j$. $\textbf{FC}_{2}^{unit}$ is the second layer of the unit-oriented stream.

Due to the concat operation, the output of the unit-oriented stream also contains the historical information provided by $h_{i, t}$, so there is no need to provide a GRU layer for it. {For the sake of brevity, the environment-oriented observations and unit-oriented observations of agent $i$ are simplified into one variable $o_{i, t}=\{o_{i, t}^{e}, o_{(i\to j), t}\}$. Finally, we concatenate all $Q_{i, t}^{e}$ and $Q_{(i\to j), t}$ to form the individual action-value function $Q_{i}(o_{i, t}, \cdot)$.

It should be noted that the individual action-value networks of other methods like QMIX fail in adapting to the change in the size of action set. They use a single network to evaluate the Q values of all kinds of actions, which include both valid and invalid actions.

\begin{algorithm}[H]
  \caption{Unshaped Networks for Multi-Agent Systems}
  \label{alg:unmas}
  \begin{algorithmic}[1]
  	\State Initialize the replay buffer $\emph{D}$ and exploration rate $\epsilon$;
  	\State Initialize $\boldsymbol{\theta}_{joint}$ and $\theta_{i}$ with random parameters;
  	\State Initialize the parameters of target networks $\theta^{-}_{i}$ and $\boldsymbol{\theta}_{joint}^{-}$: $\theta^{-}_{i} = \theta_{i}, \boldsymbol{\theta}_{joint}^{-} = \boldsymbol{\theta}_{joint}$;
	
	\For{$episode = 1 \ to \ M $}      	
    	\For{$t = 0, ..., T$}
    		\State Collect the global state $s_{i, t}$;
    		\For{each agent $i$}
    			\State Collect the local observation $o_{i, t}$;
				\State Choose a random action with probability $\epsilon$;
				\State Otherwise, estimate action-values $Q_{i}(o_{i, t}, \cdot; \theta_i)$ for each agent with the individual action-value network and choose greedy action following (\ref{eq:argmax});				
			\EndFor
			
			\State Execute the joint action $\textbf{u}_{t}$ and collect the next joint observation $\boldsymbol{o}_{t + 1}$, next state $s_{t + 1}$, and reward $r_{t + 1}$;
			\State Store the transition $(\boldsymbol{o}_{t}, \boldsymbol{u}_{t}, r_{t+1}, \boldsymbol{o}_{t+1})$ into $\emph{D}$;
			\State Store the state transition $(s_{t}, s_{t + 1})$ into $\emph{D}$;
    	\EndFor
    	\If{replay buffer $\emph{D}$ is full}
		\State Sample $b$ minibatch from $\emph{D}$;
		\State Estimate joint action-value $\textbf{Q}(\boldsymbol{o}_{k}, \boldsymbol{u}_{k}; \boldsymbol{\theta}_{joint})$ as (\ref{eq:q_jt}) by the self-weighting mixing network; 
		\State Calculate the \emph{update target} $y_k^{joint}$ in (\ref{eq:target}) and update the networks by the loss $\mathcal{L}(\theta)$ in (\ref{eq:loss});		    		
    	\EndIf
    	\State Replace target networks with \emph{target update interval};
    	
    	\State Decay the exploration rate $\epsilon$;
    \EndFor

  \end{algorithmic}
\end{algorithm}

\subsection{Training algorithm of UNMAS}
\label{sec:train}
In this subsection, we introduce the training algorithm of UNMAS. The detail is shown in Algorithm \ref{alg:unmas}. The agents take $\epsilon$-greedy policy to explore the environment in training, that is, choose a random action with probability $\epsilon$, otherwise choose the action with the highest Q value:
\begin{equation}
u_{i, t}=\arg \max_{u}Q_{i}(o_{i, t}, u; \theta_i),
\label{eq:argmax}
\end{equation}
where the exploration rate $\epsilon$ decays along the training steps. The transitions and rewards generated by the interaction between the multi-agent system and the environment are stored into the replay buffer $\emph{D}$. 

When the replay buffer $\emph{D}$ is full, we sample a minibatch from replay buffer $\emph{D}$ to update the networks. They are trained by minimizing the loss:
\begin{equation}
\mathcal{L}(\theta) = \frac{1}{b}\sum_{k=1}^{b}[(y_{k}^{joint} - \textbf{Q}(\boldsymbol{o}_{k}, \boldsymbol{u}_{k}; \boldsymbol{\theta}_{joint})]^{2},
\label{eq:loss}
\end{equation}
where $b$ is the size of minibatch. $\theta = \{\boldsymbol{\theta}_{joint}, \theta_i\}$ is the parameter of all networks involved in training. $\boldsymbol{\theta}_{joint}$ and $\theta_i$ are the parameters of self-weighting mixing network and individual action-value network respectively. $y_{k}^{joint}$ is the \emph{joint update target} of the multi-agent system, which is defined as follows:
\begin{equation}
\begin{aligned}
y^{joint}_{k} & = r_{k+1} + \gamma \textbf{Q}(\boldsymbol{o}_{k+1}, \bar{\boldsymbol{u}}_{k+1}; \boldsymbol{\theta}^{-}_{joint}) \\
\bar{\boldsymbol{u}}_{k+1} & = [\arg\max_{u_{i, k+1}}Q_{i}(o_{i, k+1}, u_{i, k+1}; \theta_{i}^{-})],
\end{aligned}
\label{eq:target}
\end{equation}
where the target action $\bar{\boldsymbol{u}}_{k+1}$ is obtained by maximizing the individual action-value function of agents. $\boldsymbol{\theta}^{-}_{joint}$ and $\theta_{i}^{-}$ are parameters of the target network for the self-weighting mixing network and the individual action-value network. Throughout the training phase, they are replaced by the current parameters every fixed episodes, which is defined as \emph{target update interval}.

The multi-agent system is treated as a whole in the training phase, so the self-weighting mixing network and individual action-value network can also be seen as one neural network, which possesses the parameter $\theta$. At each iteration, $\theta$ is updated by the gradient that is obtained by minimizing the loss function in (11). Since the self-weighting mixing network only appears in the training phase, this learning law indicates that the role of the self-weighting mixing network and its parameter is to help the update of the individual action-value network.

\begin{figure*}[htbp]
\centering
\subfloat[3m]{
\begin{minipage}{5cm}
\centering
\includegraphics[scale=0.3]{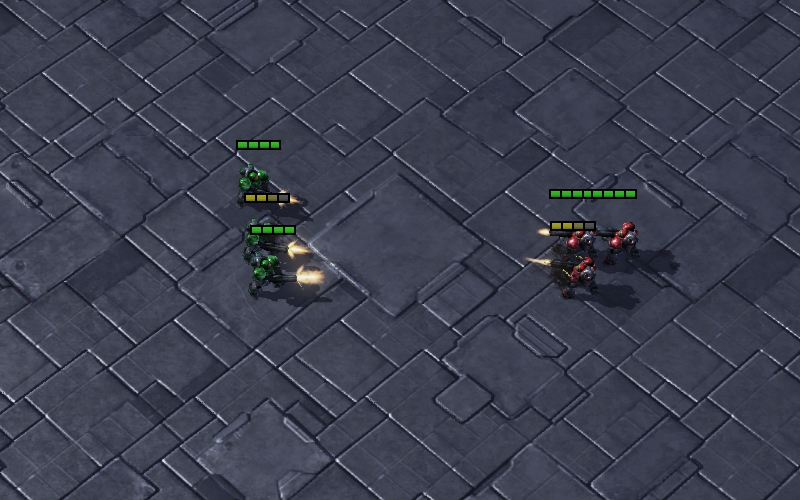}\label{fig-example}
\end{minipage}
}
\subfloat[2s3z]{
\begin{minipage}{5cm}
\centering
\includegraphics[scale=0.3]{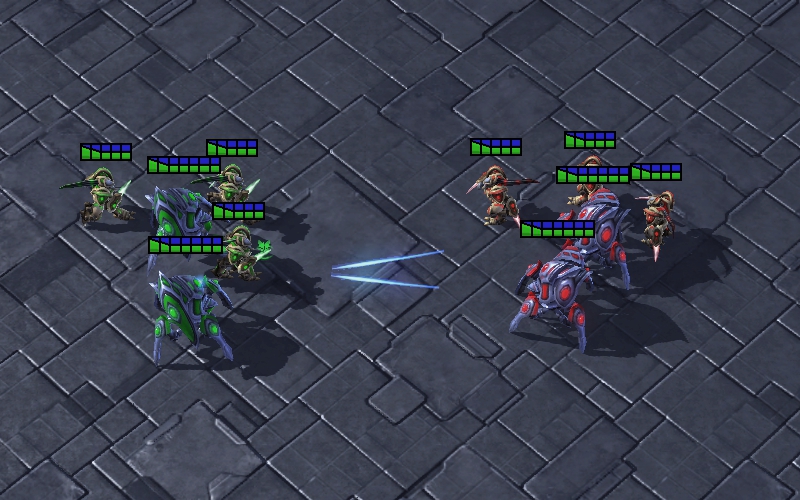}
\end{minipage}
}
\subfloat[5m\_vs\_6m]{
\begin{minipage}{5cm}
\centering
\includegraphics[scale=0.3]{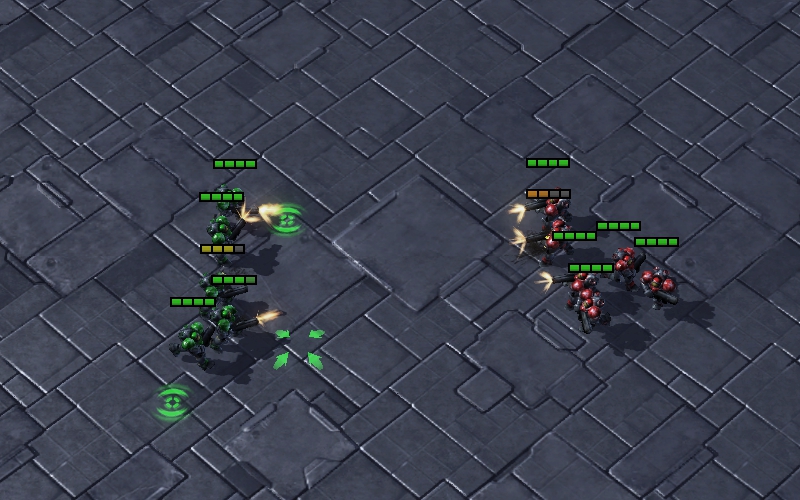}
\end{minipage}
}
\\
\subfloat[8m]{
\begin{minipage}{5cm}
\centering
\includegraphics[scale=0.3]{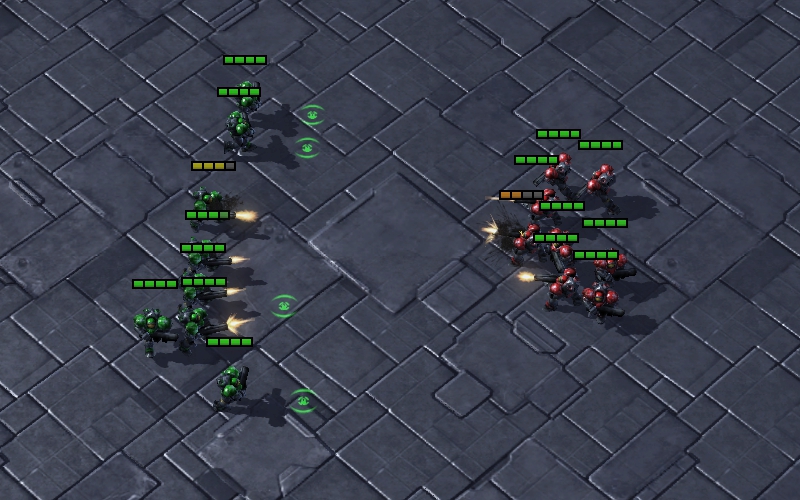}
\end{minipage}
}
\subfloat[3s5z]{
\begin{minipage}{5cm}
\centering
\includegraphics[scale=0.3]{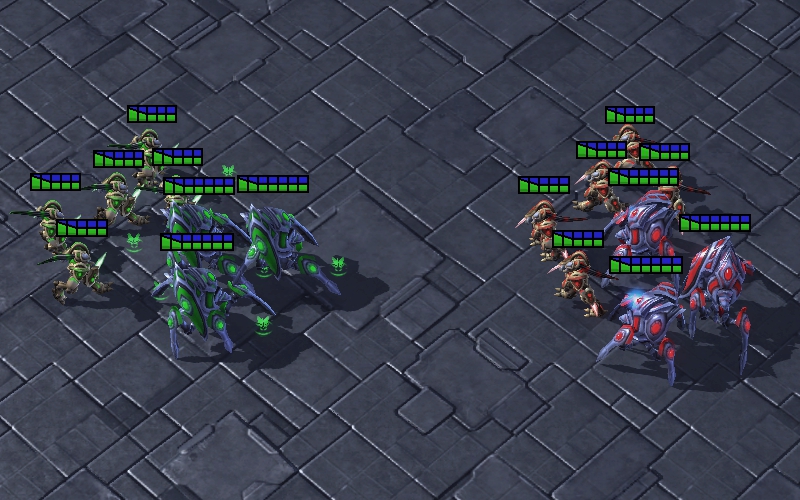}
\end{minipage}
}
\subfloat[3s5z\_vs\_3s6z]{
\begin{minipage}{5cm}
\centering
\includegraphics[scale=0.3]{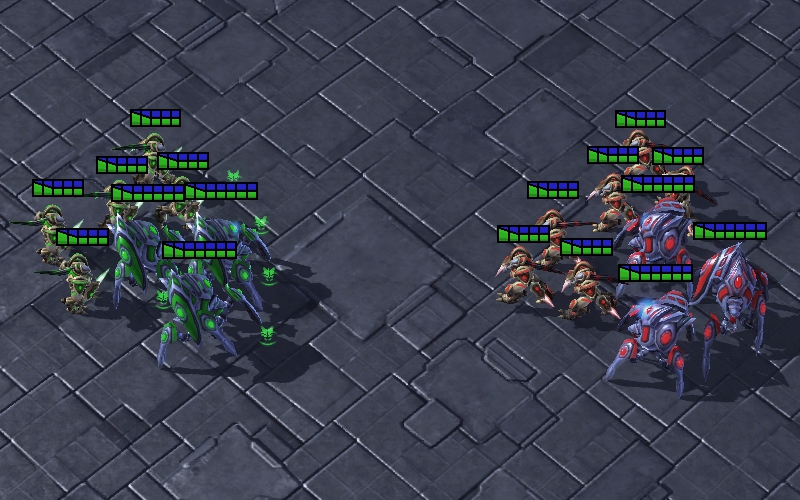}
\end{minipage}
}
\caption{The experimental scenarios to test the methods.}
\label{fig:show}
\end{figure*}

\section{Experiments on StarCraft II Micro-Management}

\subsection{Experimental Setup}
The micro-management scenario is an important aspect of StarCraft II. It requires the player to control a group of agents to win a combat against enemies. The micro-management scenario with a decentralized setting is a proper experimental scenario to test MARL methods because the agents are partially observable and can perform decentralized actions. We use StarCraft Multi-Agent Challenge (SMAC)\cite{Samvelyan2019} as the environment and evaluate our method on various maps.

We consider combats with two groups of units, one of which is using UNMAS and the other is the enemy. This scenario is an unshaped scenario, as the number of agents and enemies will decrease due to the attack. In this scenario, each agent is partially observable, which means that it could only get the information within its sight range. The observation mainly contains the following attributes for both allies and enemies: \emph{distance, relative\_x, relative\_y, health, shield}, and \emph{unit\_type}. The global state mainly contains the above information for all units, both agents and enemies. Details of the observation and state are provided in the Appendix A. Each agent has the following actions: \emph{move}[\emph{direction}], \emph{stop}, and \emph{attack}[\emph{enemy\_id}]. Each enemy is controlled by the built-in StarCraft II AI, which uses hand-crafted scripts. We set the difficulty of built-in AI as ``\emph{very difficult}''.

We take the scenario in Fig. \ref{fig:show}(a) as an example to show more details about the uncertainty in StarCraft II micro-management. In this scenario, there are three agents and three enemies at the beginning. Each agent receives two kinds of observations. One kind is environment-oriented represented by a tensor of length 42. It contains the information of agent self (8), agent allies and enemies ($2*5+3*5$), and the one-hot coding of agent last action (9). The other kind is unit-oriented represented by a collection of 3 tensors of length 5. Each tensor describes the information of an enemy (5) observed by this agent. Then, agents select actions from their action sets according to these observations. Their action sets contain the following actions: \emph{move}[\emph{up}], \emph{move}[\emph{down}], \emph{move}[\emph{left}], \emph{move}[\emph{right}], \emph{stop}, \emph{attack}[\emph{enemy\_0}], \emph{attack}[\emph{enemy\_1}], and \emph{attack}[\emph{enemy\_2}]. 

Each agent evaluates actions by the individual action-value network described before. This network takes environment-oriented observations (42) and unit-oriented observations ($3*5$) as input, and outputs the values of actions ($6+3$). The values of 6 environment-oriented actions and 3 unit-oriented actions are computed by the separate streams. If an enemy dies in combats, the size of unit-oriented observations will become $2*5$, and thus the unit-oriented stream just outputs the values of attack actions towards the remaining two enemy units. In this way, the output dimension of the individual action-value network becomes $6+2$.

During training, UNMAS samples from the replay buffer and uses the self-weighting mixing network shown in Fig. \ref{fig:self-weighting} to compute the joint action-value of the multi-agent system. The input dimension of this network is 3, which corresponds to the number of agents on the map. Furthermore, if an agent dies later in combat, the individual action-value network provides only the action-values of the remaining two agents. Therefore, the input dimension of self-weighting mixing network becomes 2. 

At each timestep, agents perform their actions in the decentralized way and receive a global reward from the environment. SMAC provides positive rewards by default, that is, the damage done by agents to enemy units. In addition, the positive (+10) reward is also provided after an enemy is killed. We evaluate the method by running 32 episodes every 10000 timesteps to get the \emph{test winning rates}. The agents use $\epsilon$-greedy to choose their action and turn it off in testing. The exploration rate $\epsilon$ decays from 1 to 0.05. 

\begin{figure*}[ht]
\centering
\subfloat{
\begin{minipage}{12cm}
\centering
\includegraphics[scale=0.35]{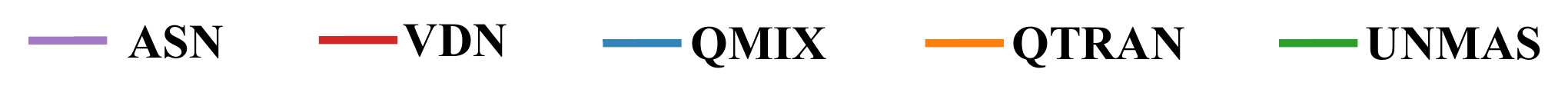}
\end{minipage}
}
\setcounter{subfigure}{0}
\subfloat[3m]{
\begin{minipage}{5cm}
\centering
\includegraphics[scale=0.325]{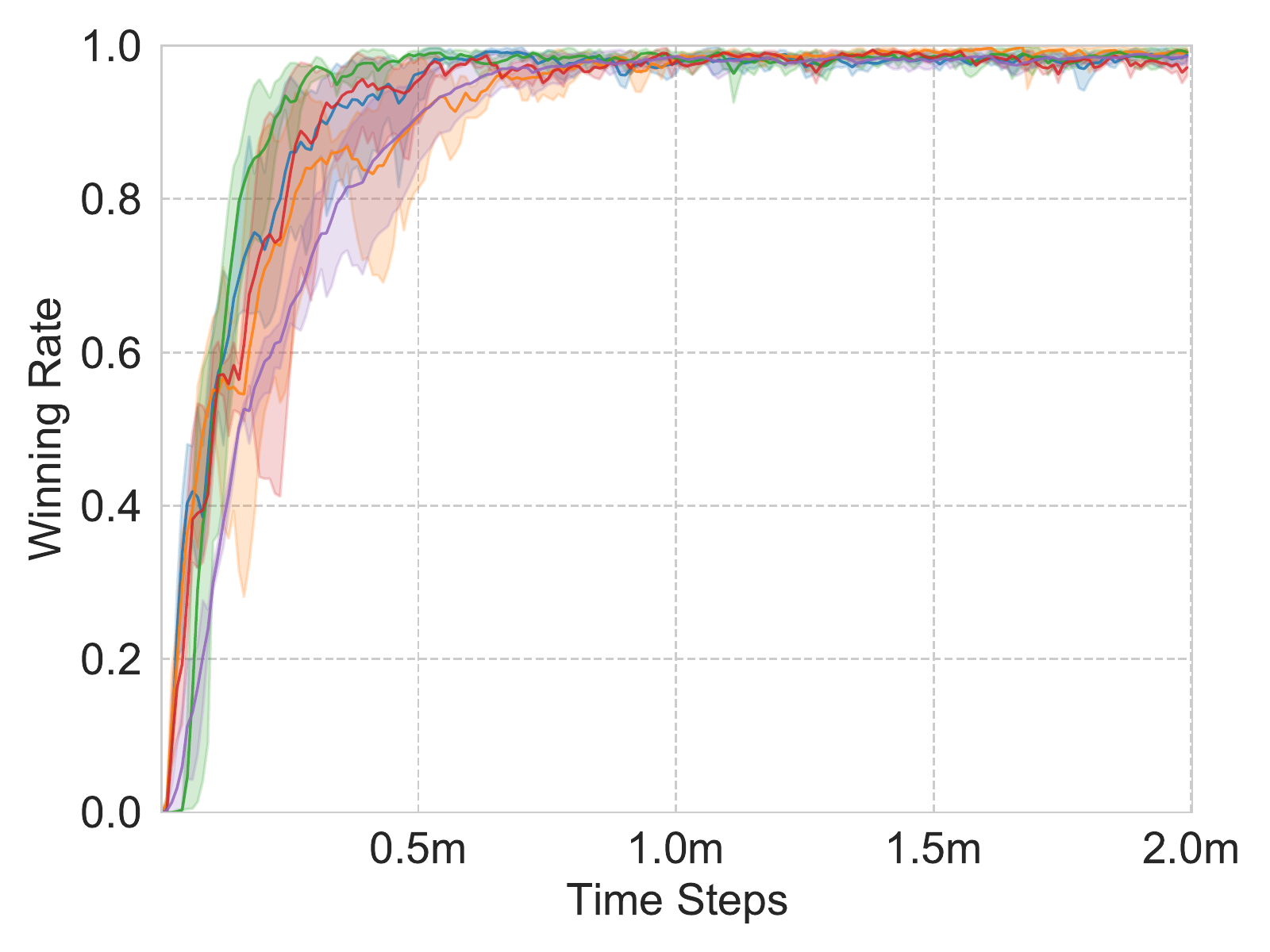}\label{fig-3m}
\end{minipage}
}
\subfloat[2s3z]{
\begin{minipage}{5cm}
\centering
\includegraphics[scale=0.325]{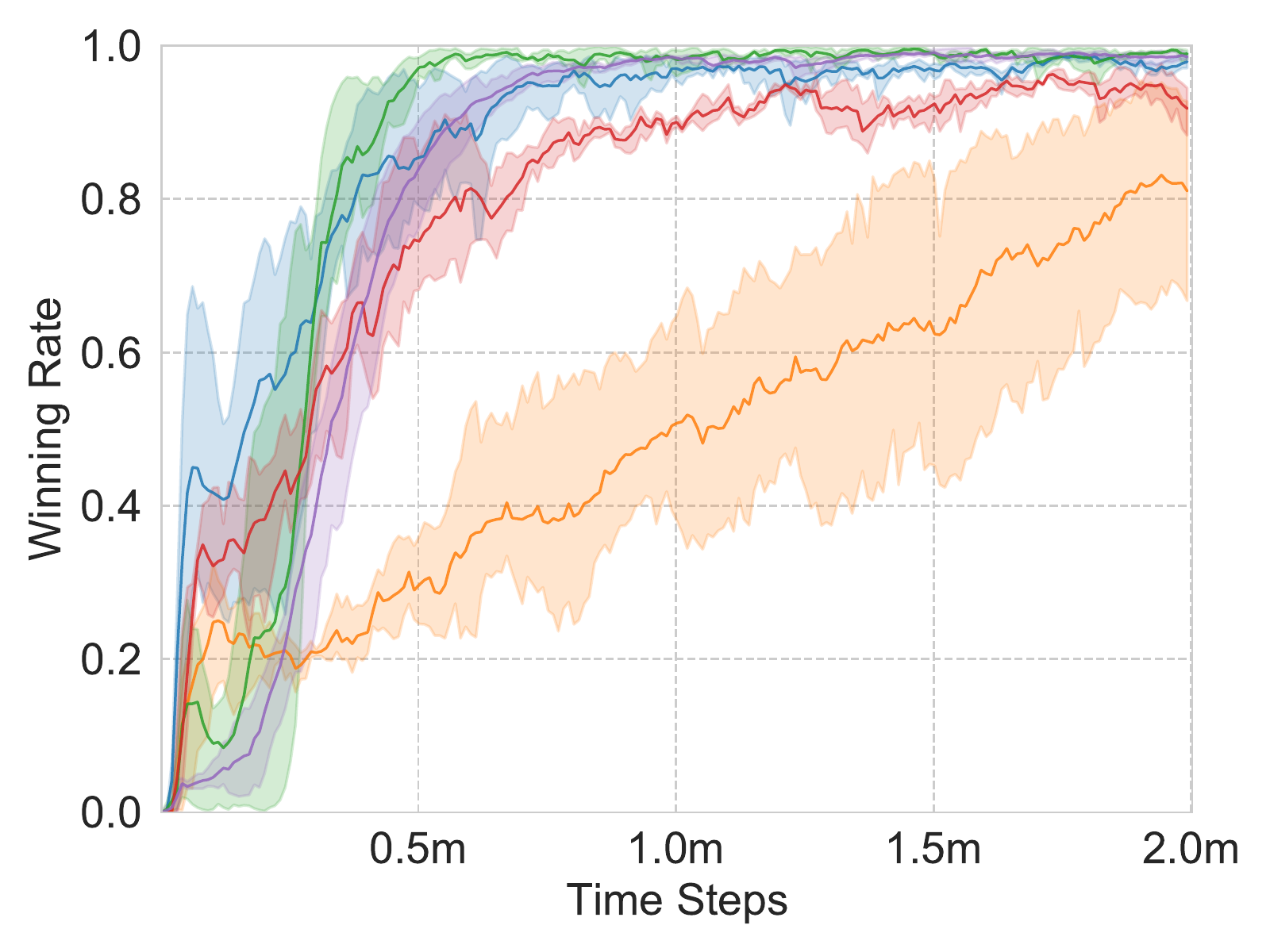}\label{fig-2s3z}
\end{minipage}
}
\subfloat[5m\_vs\_6m]{
\begin{minipage}{5cm}
\centering
\includegraphics[scale=0.325]{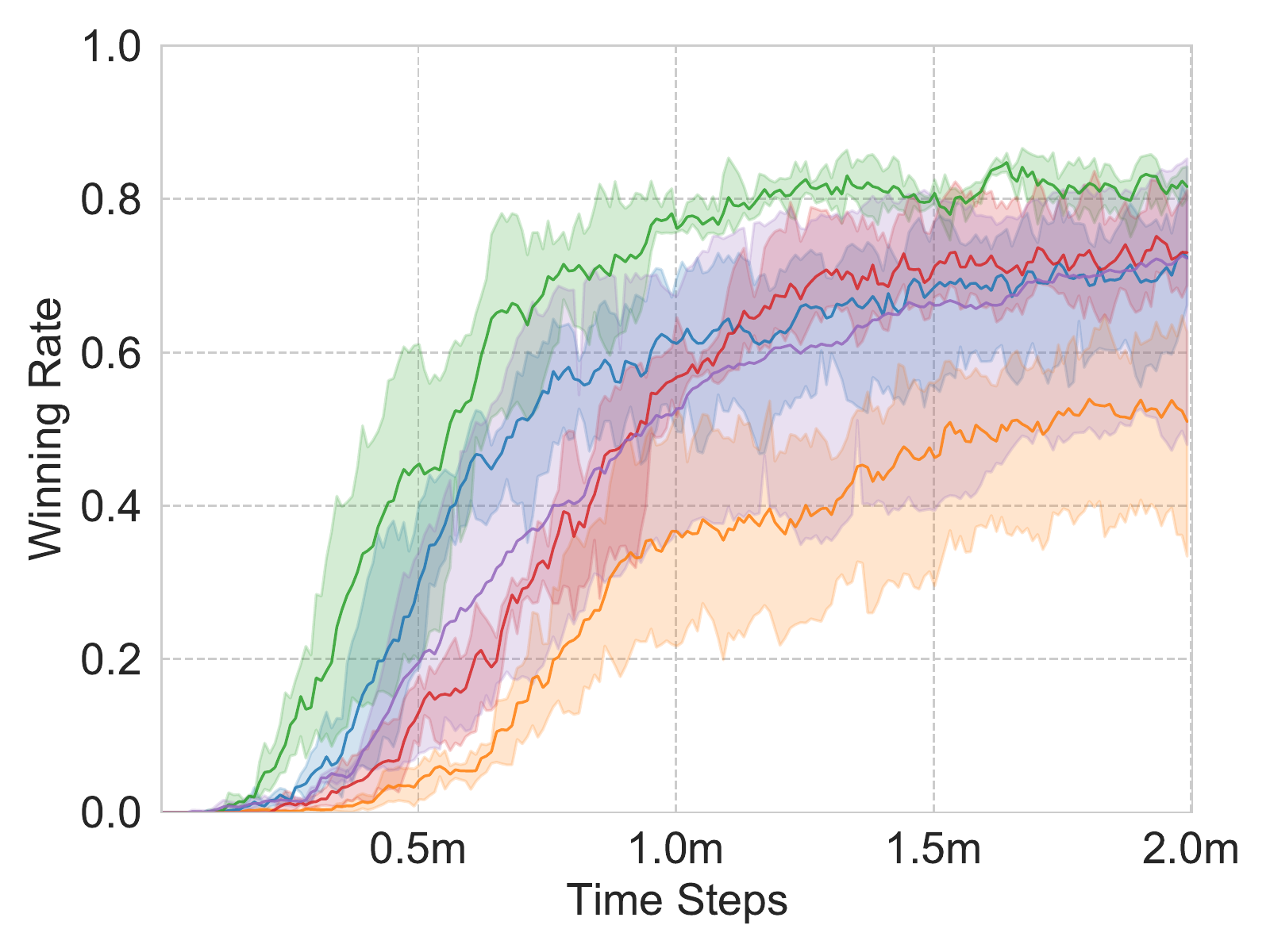}\label{fig-5v6}
\end{minipage}
}
\\
\subfloat[8m]{
\begin{minipage}{5cm}
\centering
\includegraphics[scale=0.325]{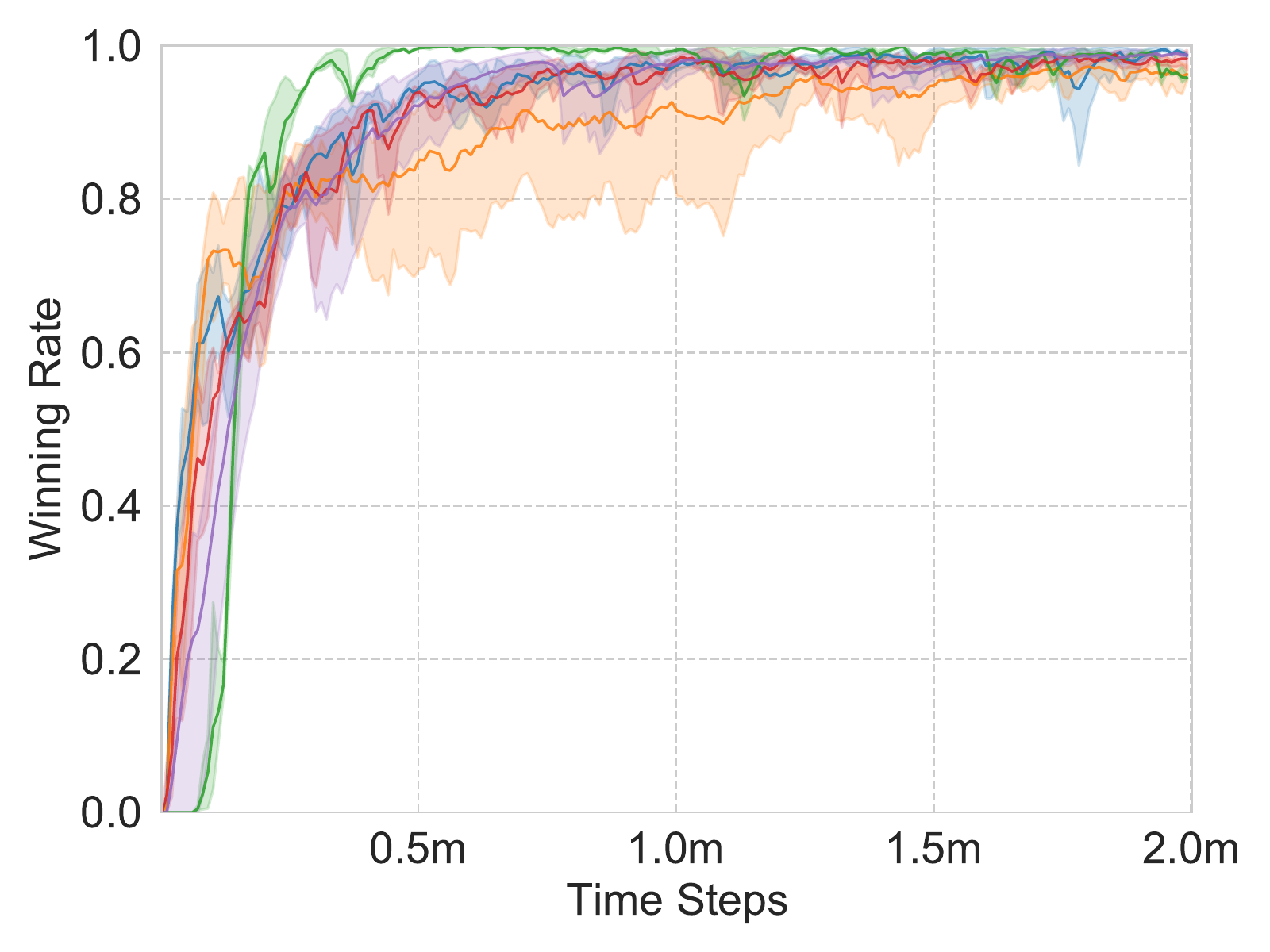}\label{fig-8m}
\end{minipage}
}
\subfloat[3s5z]{
\begin{minipage}{5cm}
\centering
\includegraphics[scale=0.325]{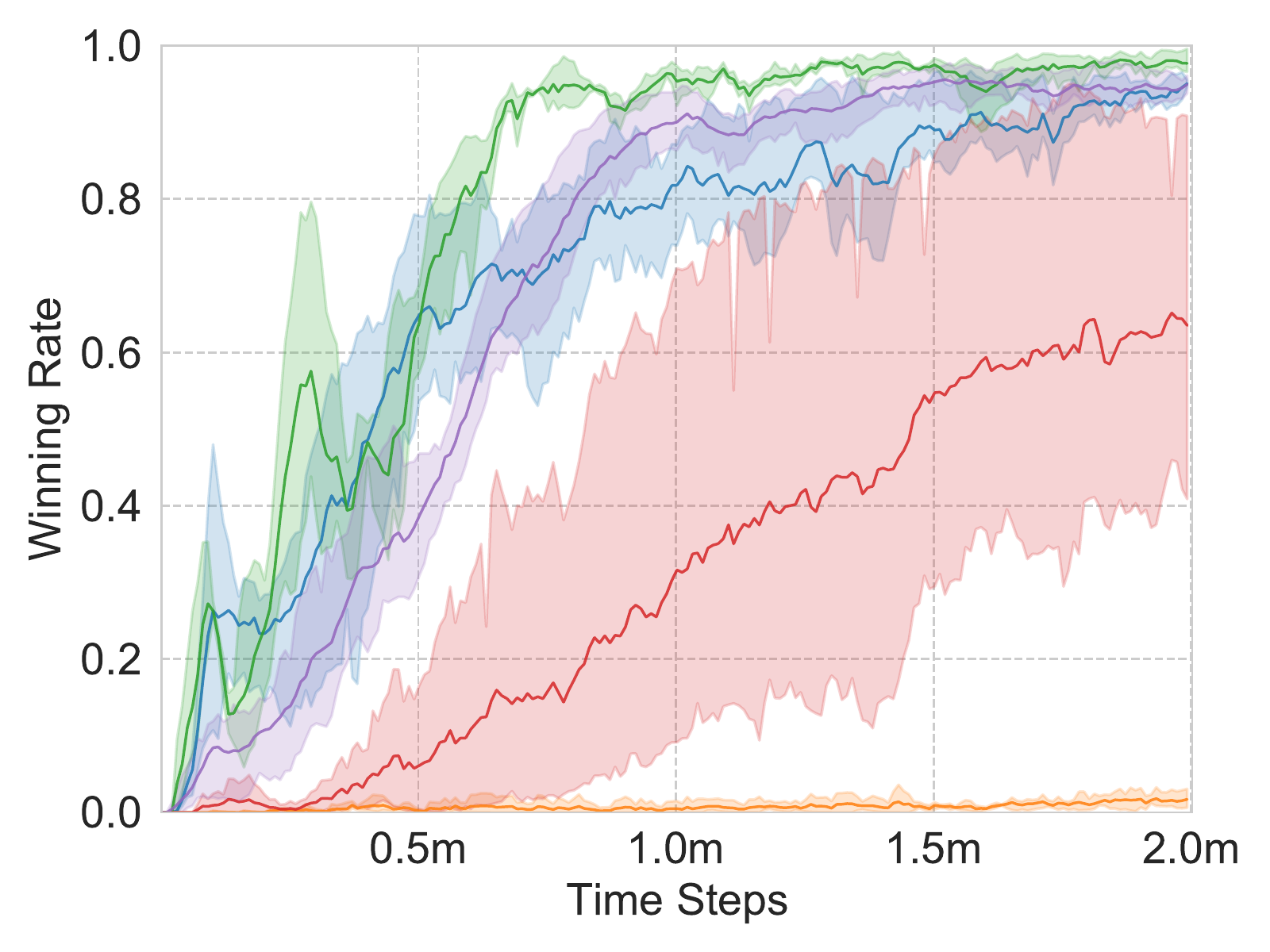}\label{fig-3s5z}
\end{minipage}
}
\subfloat[3s5z\_vs\_3s6z]{
\begin{minipage}{5cm}
\centering
\includegraphics[scale=0.325]{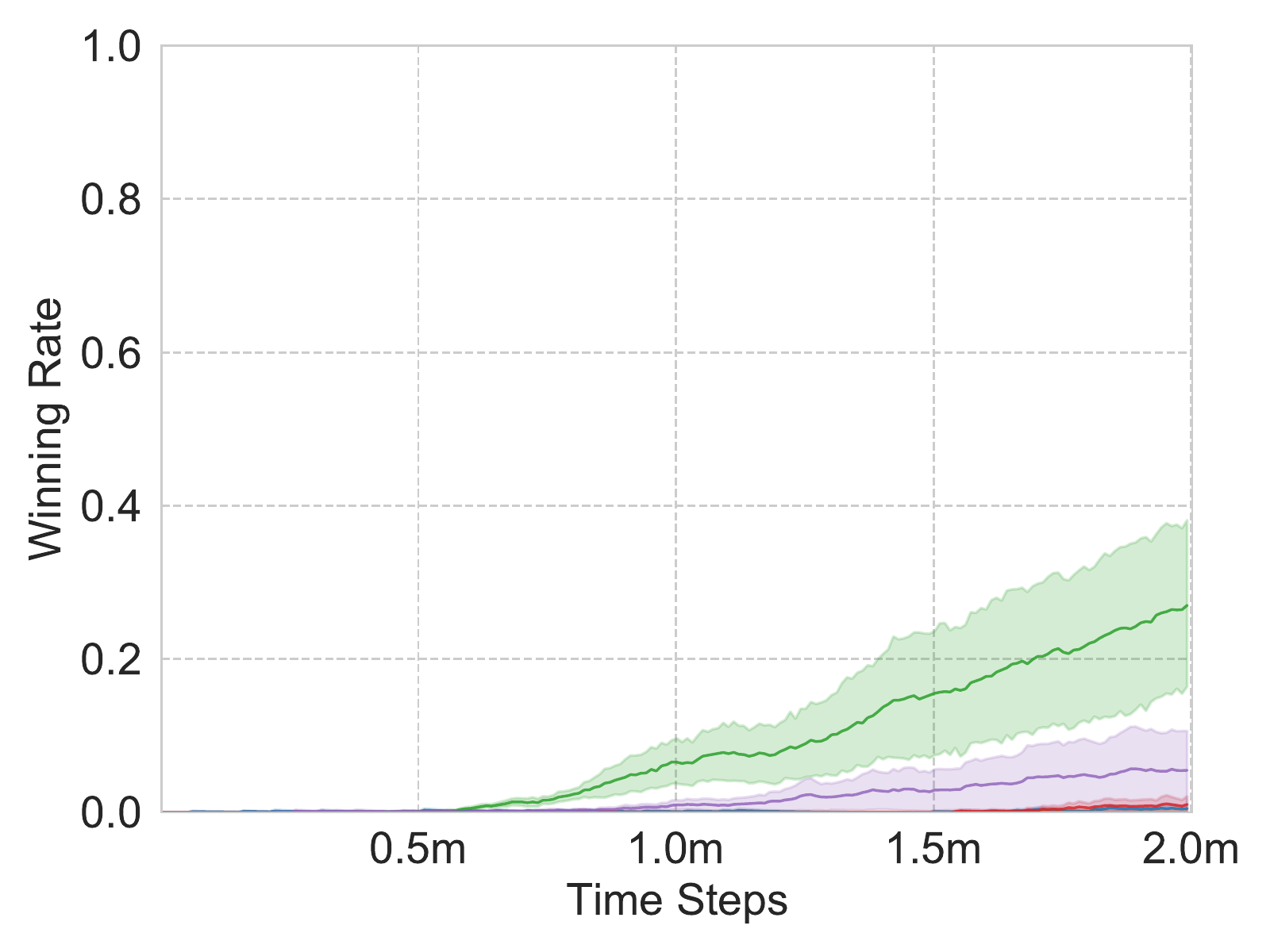}\label{fig-8v9}
\end{minipage}
}
\caption{Winning rates for ASN, VDN, QMIX, QTRAN and UNMAS. UNMAS achieves the highest winning rate in all scenarios, especially the most difficult scenario $3s5z\_vs\_3s6z$.}
\label{fig:win-rate}
\end{figure*}

We compare UNMAS with VDN \cite{Sunehag2018}, QMIX \cite{Rashid2018}, QTRAN \cite{Son2019}, and ASN \cite{wang2020action}. VDN factorizes the joint action-value $\textbf{Q}(\textbf{o}_{t}, \textbf{u}_{t})$ by the summation of the individual action-value function $Q_{i}(o_{i, t}, u_{i, t};\theta_{i})$:
\begin{equation}
\textbf{Q}(\textbf{o}_{t}, \textbf{u}_{t}) = \sum_{i=1}^{n}Q_{i}(o_{i, t}, u_{i, t};\theta_{i})
\end{equation}
where $\theta_{i}$ is the parameter of the newtork representing the individual action-value function. QMIX employs a mixing network to represent the joint action-value function. This network takes the action-value $Q_{i}(o_{i, t}, u_{i, t};\theta_{i})$ of all agents as input and the joint action-value $\textbf{Q}(\textbf{o}_{t}, \textbf{u}_{t}; \boldsymbol{\theta}_{joint})$ as output, where $\boldsymbol{\theta}_{joint}$ is its parameter. By being taken absolute values, the weights of the mixing network keep non-negative. QTRAN factorizes the joint action-value by designing a special learning objective. Therefore, it can transform the original joint action-value function into a new, easily factorizable one with the same optimal actions in both functions. ASN considers the action semantics between agents to compute the action-values and reconstruct the observation of agent $i$ at timestep $t$ as follows: $o_{i, t}=\{o_{i, t}^{env}, m_{i, t}, o_{i, t}^1, ...,o_{i, t}^{i-1}, o_{i, t}^{i+1},o_{i, t}^n\}$, where $o_{i, t}^{env}$ is the observation about environment, $m_{i, t}$ is the private information about agent $i$ itself, and $o_{i, t}^j$ is the observation of agent $i$ to other agent $j$. 

Experiments are performed on the following symmetric maps, where agents and enemies have the same numbers and types: 3 Marines ($3m$), 8 Marines ($8m$), 2 Stalkers and 3 Zealots ($2s3z$), 3 Stalkers and 5 Zealots ($3s5z$). Besides, the asymmetric maps including 5 Marines against 6 Marines ($5m\_vs\_6m$), 3 Stalkers and 5 Zealots against 3 Stalkers and 6 Zealots ($3s5z\_vs\_3s6z$) are also chosen for comparison. The description of these units is provided in Table \ref{table:des}.

\begin{table}[ht]
\caption{Description of StarCraft II units.}
\label{table:des}
\renewcommand{\arraystretch}{1.1}
\setlength{\tabcolsep}{7mm}
\begin{center}
\begin{tabular}{ccc}
\toprule
Unit & Race & Type \\ \midrule
Marine & Terran & Long-range \\
Stalker & Protoss & Long-range  \\
Zealot & Protoss & Short-range  \\
\bottomrule
\end{tabular}
\end{center}
\end{table}

These maps are provided by SMAC and could test the performance of UNMAS on StarCraft II micro-management scenario. VDN \cite{Sunehag2018}, QMIX \cite{Rashid2018}, and QTRAN \cite{Son2019} are commonly used for comparison in CTDE methods and have already been implemented in SMAC. The code of ASN \cite{wang2020action} is also open-source. Besides, the detailed implementation of networks and hyper parameters are provided in Appendix Table \ref{table:hp}.

\begin{table*}[ht]
\caption{The test winning rates of VDN, QMIX, QTRAN, ASN, UNMAS, and other state-of-the-art methods.}
\label{table:wr}
\renewcommand{\arraystretch}{1.1}
\begin{center}
\begin{tabular}{ccccccccccc}
\toprule
Map 				& VDN$\dagger$	& QMIX$\dagger$ 	& QTRAN$\dagger$	& ASN$\dagger$	& RODE$\dagger$	& QPD$\ddagger$	& MADDPG$\ddagger$	& Qatten$\ddagger$	& QPLEX$\ddagger$	& UNMAS 		\\ \midrule
3m 				& 98 			& \textbf{99}		& \textbf{99}		& 98				& \textbf{99}		& 92			& 98   			& - 				& - 				& \textbf{99} 	\\ 
8m 				& 97 			& \textbf{98}		& 97				& \textbf{98}		& \textbf{98}		& 93			& \textbf{98}		& - 				& - 				& 97		  	\\ 
2s3z 			& 92 			& 98				& 81 			& \textbf{99}		& \textbf{99}  	& \textbf{99}	& 94				& 96 				& 97 				& 96 			\\
3s5z 			& 63 			& 95				& 2				& 95				& 92 			& 91			& 72				& 94 				& 92				& \textbf{98}  	\\
5m\_vs\_6m 		& 73 			& 72				& 55				& 72				& 75				& -			& - 				& 74 				& - 				& \textbf{82}  	\\
3s5z\_vs\_3s6z 	& 1 				& 1				& 0				& 5				& 1 				& 10			& -				& 16 				& - 				& \textbf{28}  	\\
\bottomrule
\end{tabular}
\begin{tablenotes}
     \item[1] $\dagger$ Obtained by the experiments that we conduct. 
     \item[2] $\ddagger$ Obtained by the results provided in literature.
\end{tablenotes}
\end{center}
\end{table*}

\subsection{Main Results}
We apply these methods on the maps, and choose the test winning rate as the comparison metric. Each experiment is repeated three times for average results, and the learning curves are shown in Fig. \ref{fig:win-rate}.

\begin{enumerate}[labelsep = .5em, leftmargin = 0em, itemindent = 2em]
\item[1)] \emph{Symmetric Scenarios}: The maps with symmetric setting include: $3m$, $8m$, $2s3z$ and $3s5z$. 

\begin{enumerate}[leftmargin = 2em]
\item[a)] \emph{Homogeneous Scenarios}: $3m$ and $8m$. In these scenarios, there is only one type of unit: Marine. Therefore, the individual action-value network only needs to represent the policy of one type of unit. As shown in Fig. \ref{fig:win-rate}(a) and Fig. \ref{fig:win-rate}(d), all of the methods quickly achieve a winning rate close to 100\%. The curve of UNMAS starts rising slowly, but still achieves a higher winning rate faster than other methods. 

\item[b)] \emph{Heterogeneous Scenarios}: $2s3z$ and $3s5z$. In these scenarios, there are two types of units: Zealot and Stalker, so the individual action-value network needs to represent the policies of different types of units. Different types of units have different roles in combat, so heterogeneous scenarios are more difficult than homogeneous scenarios. As shown in Fig. \ref{fig:win-rate}(b) and Fig. \ref{fig:win-rate}(e), UNMAS achieves the highest winning rate on these maps, especially $3s5z$. The $3s5z$ map is more difficult than $2s3z$, because the increase in the number of units makes the multi-agent system harder to control. In the $3s5z$ experiment, VDN and ASN can not perform well, and QTRAN fails to win. Only UNMAS and QMIX can win combats stably.

\end{enumerate}

\item[2)] \emph{Asymmetric Scenarios}: The maps with asymmetric setting include: $5m\_vs\_6m$ and $3s5z\_vs\_3s6z$. It means that the number of enemies is greater than that of own agents.

\begin{enumerate}[leftmargin = 2em]
\item[a)] \emph{Homogeneous Scenarios}: $5m\_vs\_6m$. In this scenario, 5 Marines controlled by the algorithms have to combat against 6 enemy Marines. Even if there is only one difference in the number, the control difficulty in order to win increases considerably. Therefore, the policies of agents are required to be more elaborate in order to win. A more elaborate policy means that the agents make fewer mistakes. As shown in Fig. \ref{fig:win-rate}(c), no method achieves 100\% winning rate. Among the tested methods, UNMAS achieves the highest winning rate, which means that the multi-agent system using UNMAS is more possible to achieve cooperation and choose correct actions.

\item[b)] \emph{Heterogeneous Scenarios}: $3s5z\_vs\_3s6z$. It is an asymmetric and heterogeneous scenario, the most difficult of all. With an extra Zealot, the enemy is better able to block our Zealots and attack our important \emph{damage dealer} Stalkers. Therefore, the multi-agent system not only needs to ensure the stability of its policy but also explores a better policy to win. As shown in Fig. \ref{fig:win-rate}(f), UNMAS achieves the highest winning rate. 

\end{enumerate}

\end{enumerate}

Throughout the plot in Fig. \ref{fig:win-rate}, it is shown that UNMAS achieves uniform convergence on all experimental maps, and is especially superior on the most difficult scenario 3s5z\_vs\_3s6z. The final winning rates of different methods are listed in Table \ref{table:wr}. Compared with QMIX and other methods, we provide a more flexible network structure to represent the joint action-value function, so it is more suitable for training on unshaped scenarios. As for the design of the individual action-value network, ASN uses the GRU layer in unit-oriented stream, while we remove it. Since we are considering a Dec-POMDP especially with limited sight range, a unit wandering close to an agent's sight range will cause it to receive the unit’s observations intermittently. Therefore, it may lead to a wrong computation of the hidden state for GRU layer and have a negative effect on the strategy of agent.

Besides, for other state-of-the-art MARL methods, due to the limits of reproduction, we only compare with the results provided in their literature, and show them in Table \ref{table:wr}. MADDPG \cite{NIPS2017_7217} provides a centralized critic for the whole system and a group of decentralized actors for each agent. Qatten \cite{Yang2020b} designs a multi-head attention network to approximate joint action-value function. RODE \cite{Wang2021rode} allows agents with a similar role to share similar behaviors. QPD \cite{Yang2020a} employs integrated gradients method to factorize the joint action-value of multi-agent system into individual action-values of agents. QPLEX \cite{Wang2020a} takes a duplex dueling network to factorize the joint action-value. 

The results shown in Table \ref{table:wr} are the average winning rates after 2 million timesteps. In the article of RODE, the experiment on most maps is the result of training 5 million timesteps, we run RODE on each map and test its winning rate at 2 million timesteps. Since MADDPG does not experiment on StarCraft II micro-management scenarios, we compare with the results provided by \cite{papoudakis2020comparative}, in which the authors implement several MARL methods. Although these methods achieve similar performance to us in symmetric maps, UNMAS still has obvious advantage in the most difficult scenario 3s5z\_vs\_3s6z. These results indicate that UNMAS is still competitive among those state-of-the-art methods. 

\begin{figure*}[ht]
\centering
\subfloat{
\begin{minipage}{12cm}
\centering
\includegraphics[scale=0.32]{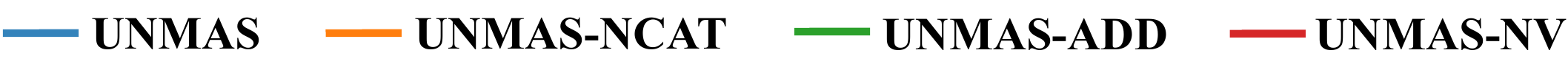}
\end{minipage}
}
\setcounter{subfigure}{0}
\subfloat[3s5z]{
\begin{minipage}{5cm}
\centering
\includegraphics[scale=0.325]{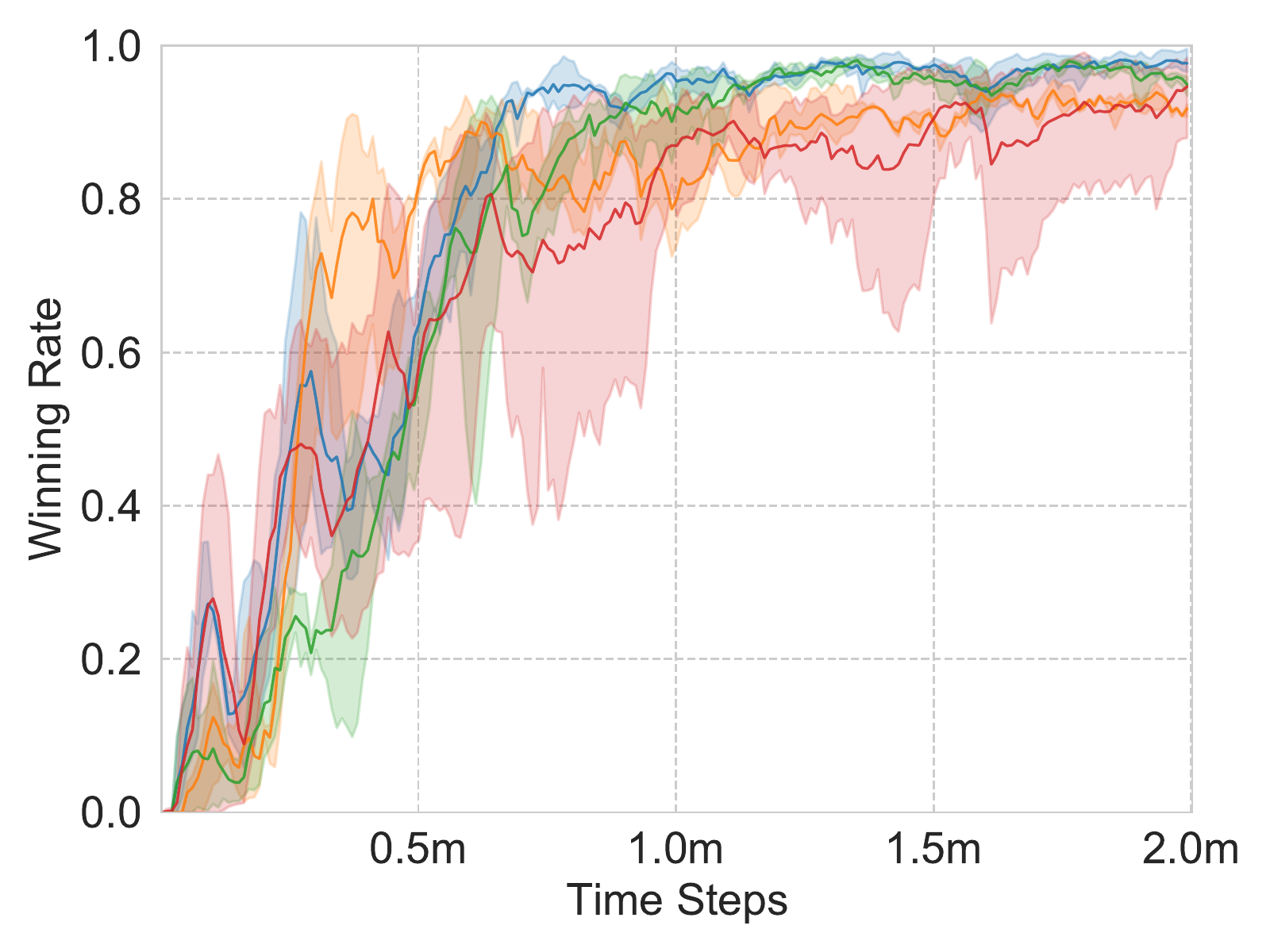}
\end{minipage}
}
\subfloat[5m\_vs\_6m]{
\begin{minipage}{5cm}
\centering
\includegraphics[scale=0.325]{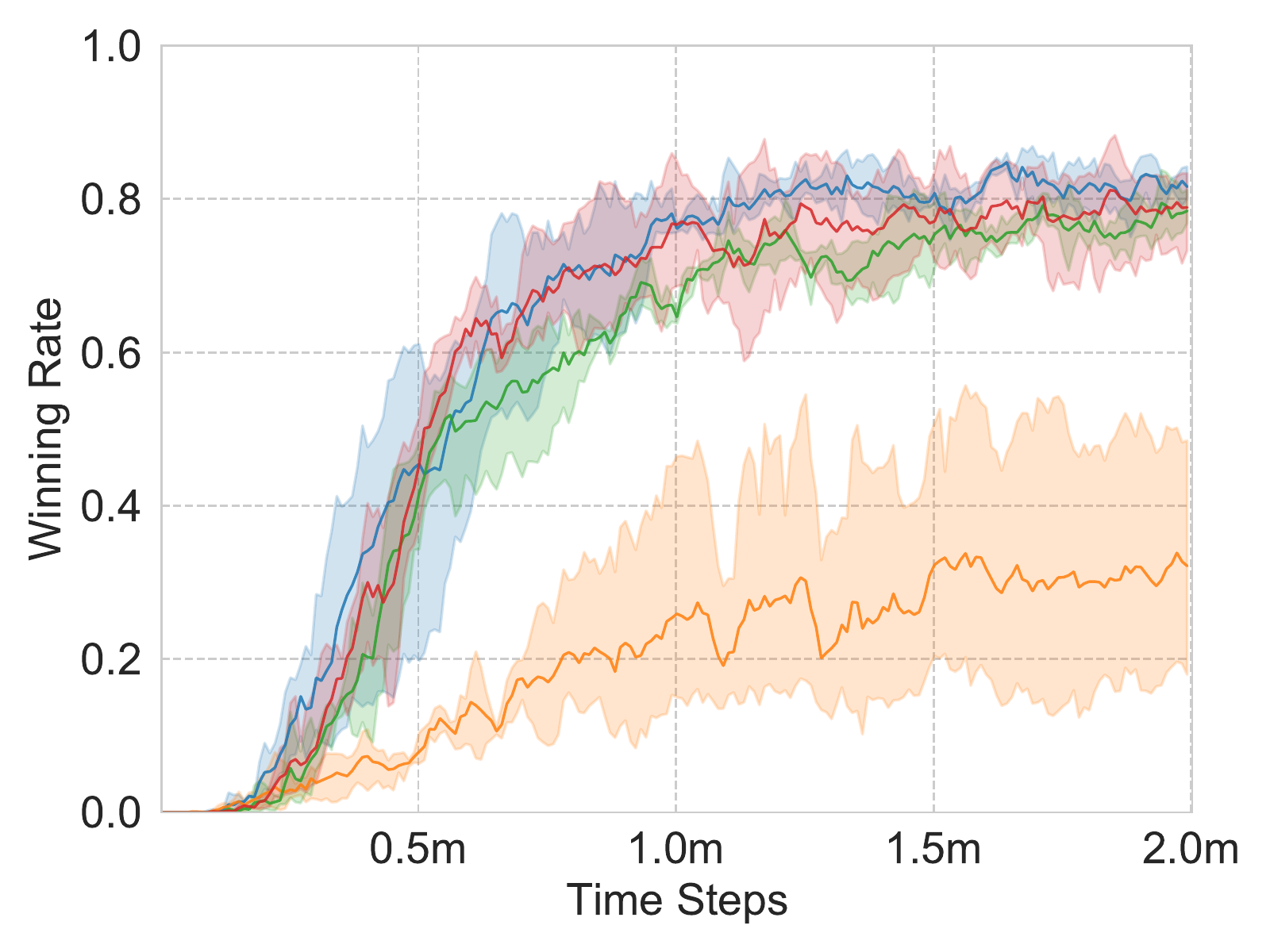}
\end{minipage}
}
\subfloat[3s5z\_vs\_3s6z]{
\begin{minipage}{5cm}
\centering
\includegraphics[scale=0.325]{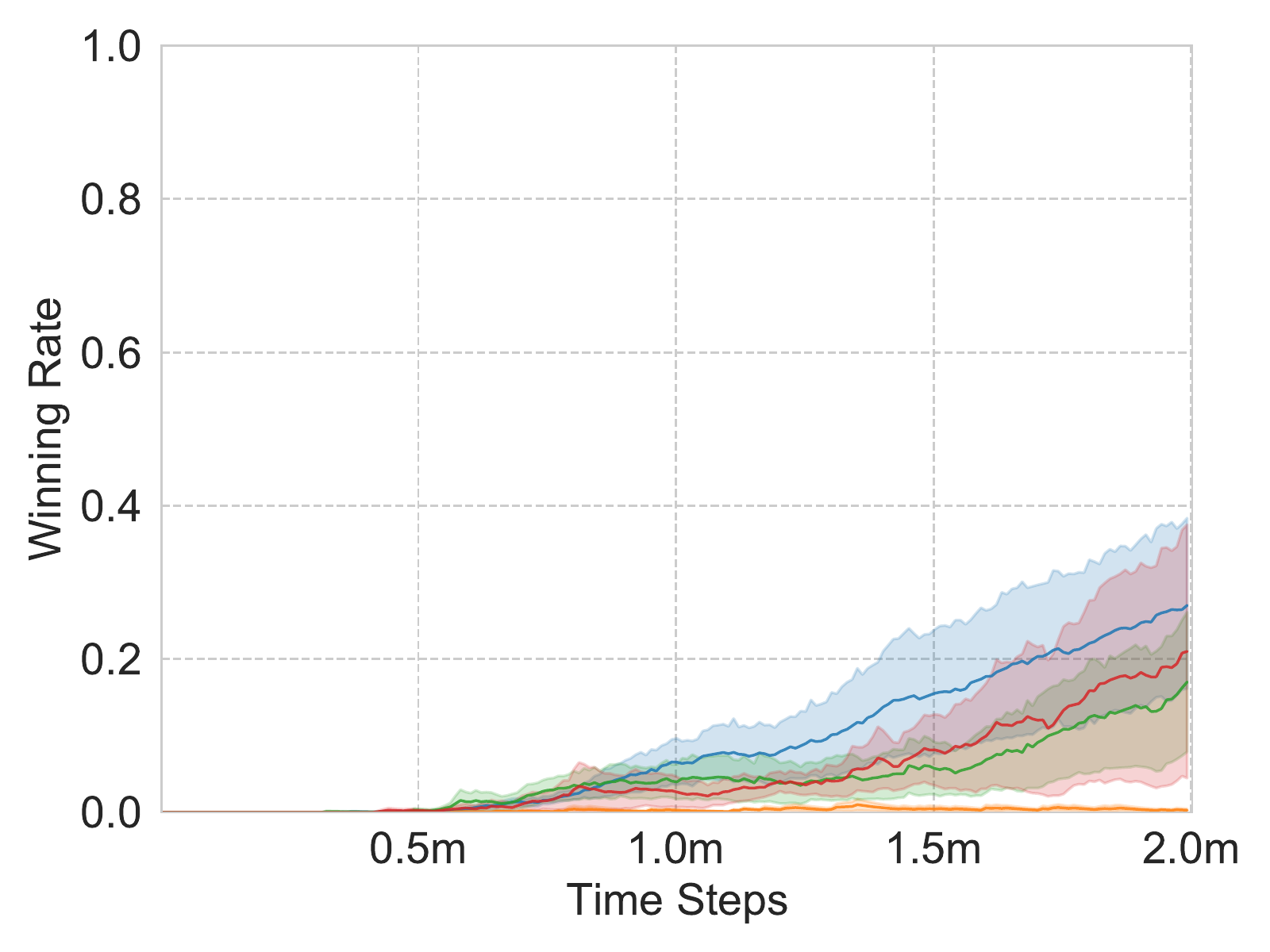}
\end{minipage}
}
\caption{Winning rates of ablation experiments. The ablation results demonstrate the effect of three elements.}
\label{fig:ablation}
\end{figure*}

\subsection{Ablation Results}
In order to investigate the effect of (i) the weight $k_{i, t}$, (ii) the value $v_{t}$ in the self-weighting mixing network, and (iii) the concat operation in the individual action-value network, we conduct ablation experiments on three maps: $3s5z$, $5m\_vs\_6m$ and $3s5z\_vs\_3s6z$.

\begin{enumerate}
\item[1)] \emph{The effect of weight} $k_{i, t}$. We remove the weight $k_{i, t}$ in the self-weighting mixing network and replace the joint action-value with the following equation:
\begin{equation}
\textbf{Q}(\textbf{o}_{t}, \textbf{u}_{t}; \boldsymbol{\theta}_{joint}) = \frac{1}{n} \sum_{i=1}^{n} q_{i, t}^{'} + v_{t},
\end{equation}
which denotes that the weight $k_{i, t}$ of each agent is set to 1 in any case. We refer to this method as UNMAS-ADD. Fig. \ref{fig:ablation} shows that the winning rate of self-weighting mixing network decreases under the condition of fixed weights $k_{i, t}$. One possible explanation is that because the weights are fixed, the self-weighting mixing network cannot correctly estimate the true joint action-value based on the contribution of each agent to the joint action-value. 

\begin{figure}[h]
\centering
\subfloat[Before Fighting]{
\begin{minipage}{4.25cm}
\centering
\includegraphics[scale=0.265]{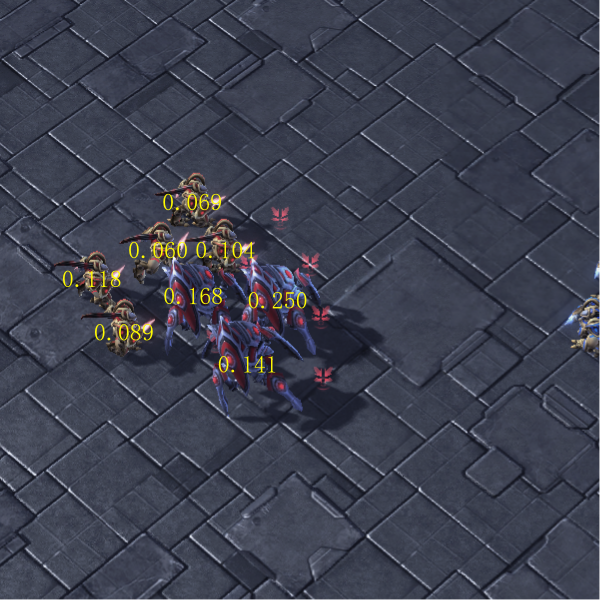}\label{fig-before}
\end{minipage}
}
\subfloat[Fighting Hand-to-hand]{
\begin{minipage}{4.25cm}
\centering
\includegraphics[scale=0.25]{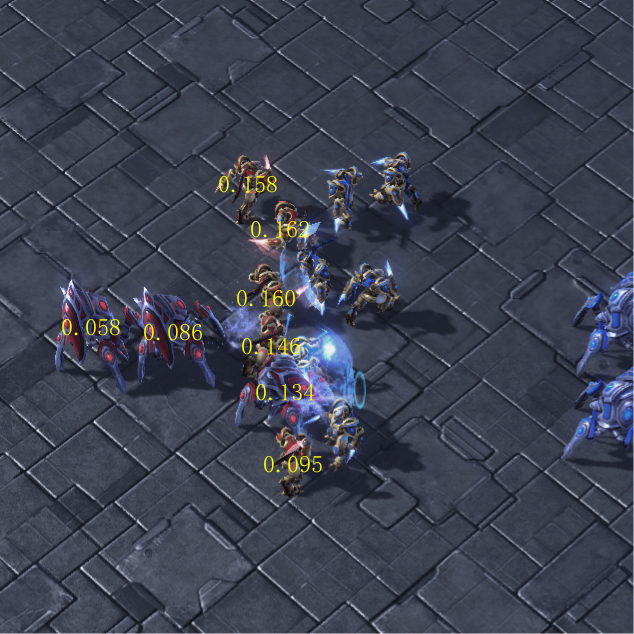}\label{fig-hand}
\end{minipage}
}
\caption{Two screenshots of the combat on the $3s5z\_vs\_3s6z$ map. The weight ratio $k_{i, t} / \sum_j k_{j, t}$ is marked on the position of each agent $i$.}
\label{fig:weight}
\end{figure}

In order to illustrate the effect of $k_{i, t}$ specifically, we provide two screenshots of the combat on the $3s5z\_vs\_3s6z$ map and mark the weight ratio $k_{i, t} / \sum_j k_{j, t}$ of each agent as shown in Fig. \ref{fig:weight}. The screenshots indicate that agents have different contribution to the joint action-value. When there is no close combat between two sides as shown in Fig. \ref{fig:weight}(a), the long-range unit Stalker is more important. Its average weight of 0.186 is larger than Zealot's 0.073. Once they are fighting hand-to-hand as shown in Fig. \ref{fig:weight}(b), the short-range unit Zealot becomes vital. Its average weight of 0.120 is larger than Zealot's 0.092. Besides, within the same type of units, agents get larger weights if there are more enemies around them. Taking the Zealots in Fig. \ref{fig:weight}(b) as an example, the unit surrounded by enemies has larger weight (0.162) than the one that is alone (0.095).

\item[2)] \emph{The effect of value} $v_{t}$. In the second experiment, we remove the value $v_t$ in the joint action-value. Therefore, the joint action-value can be calculated by the following formula:
\begin{equation}
\textbf{Q}(\textbf{o}_{t}, \textbf{u}_{t}; \boldsymbol{\theta}_{joint}) = \frac{1}{n} \sum_{i=1}^{n} q_{i, t}^{'} \cdot k_{i, t}.
\end{equation}
We refer to this method as UNMAS-NV. Fig. \ref{fig:ablation} shows that the introduction of $v_{t}$ increases the efficiency of approximation, which leads to higher winning rate. It indicates that the lack of a bias term makes it more difficult for the self-weighting mixing network to approximate the joint action-value function.

\item[3)] \emph{The effect of the concat operation}. In the third experiment, we remove the concat operation in the individual action-value network. It means that the Q values used to evaluate unit-oriented actions only depend on the observations of target units rather than the observation history. We refer to this method as UNMAS-NCAT. Fig. \ref{fig:ablation} shows that the concat operation is critical to the performance of the agent. The winning rate of UNMAS-NCAT in ablation experiments is much lower than UNMAS.

\end{enumerate}

According to the ablation experiments, we figure out the importance of the weight $k_{i, t}$, value $v_{t}$, and the concat operation. The reason that UNMAS achieves the highest winning rate becomes clear. In the factorization of joint action-value, UNMAS provides the weights for each alive agent and ignores the agents that are killed by enemies, which leads to a more proper factorization. The bias term also helps approximate the joint action-value function. In estimating the action-value of an agent, UNMAS evaluates the actions in two different subsets respectively with the help of the concat operation, making the evaluation more accurate.

\begin{figure*}[ht]
\centering
\subfloat[More to Fight Less]{
\begin{minipage}{4.25cm}
\centering
\includegraphics[scale=0.276]{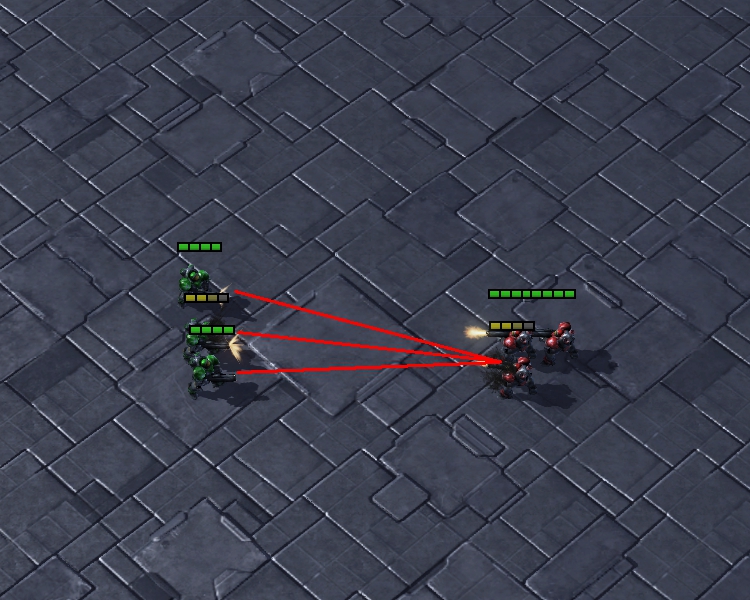}\label{fig-mtol}
\end{minipage}
}
\subfloat[Damage Sharing]{
\begin{minipage}{4.25cm}
\centering
\includegraphics[scale=0.276]{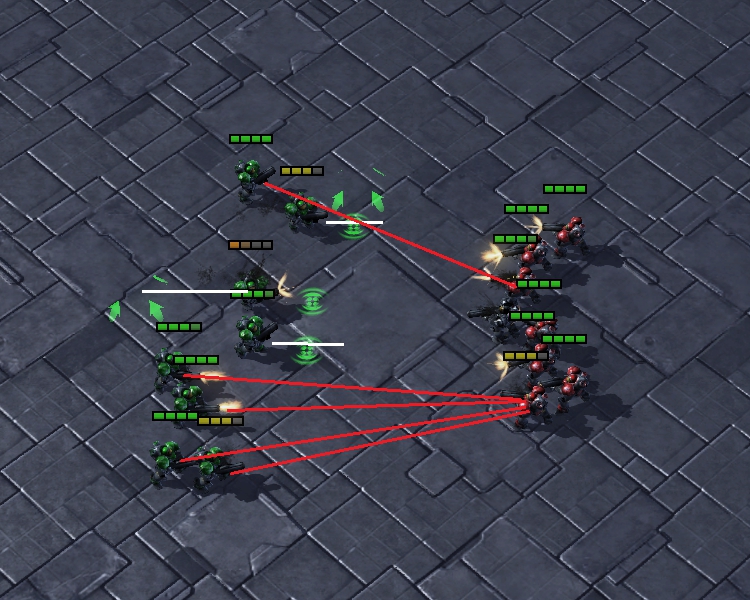}\label{fig-damage}
\end{minipage}
}
\subfloat[Block Enemies]{
\begin{minipage}{4.25cm}
\centering
\includegraphics[scale=0.23]{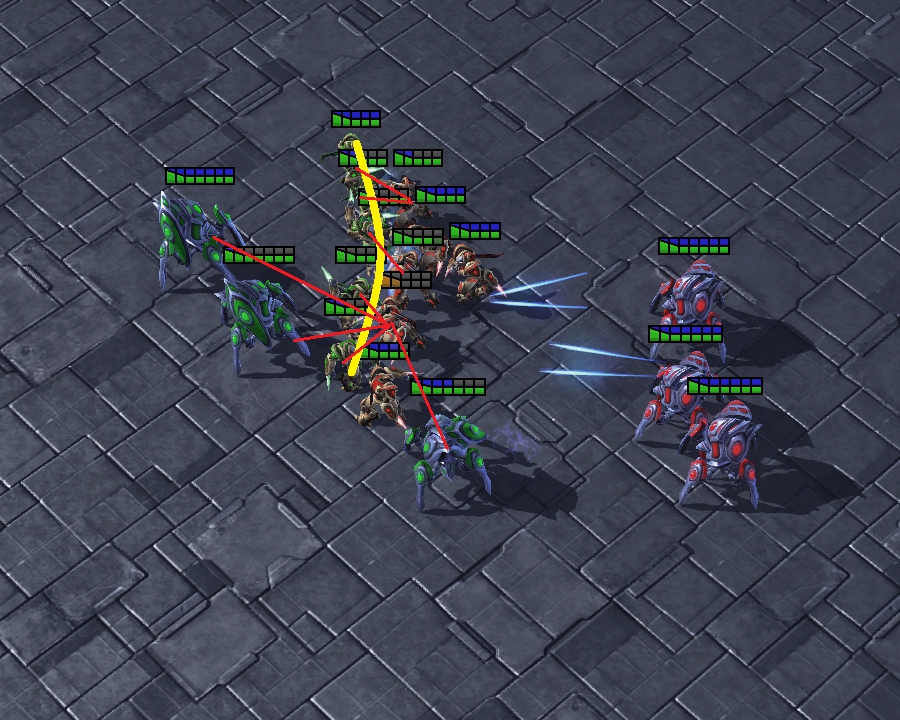}\label{fig-block}
\end{minipage}
}
\subfloat[Defensive Counterattack]{
\begin{minipage}{4.25cm}
\centering
\includegraphics[scale=0.23]{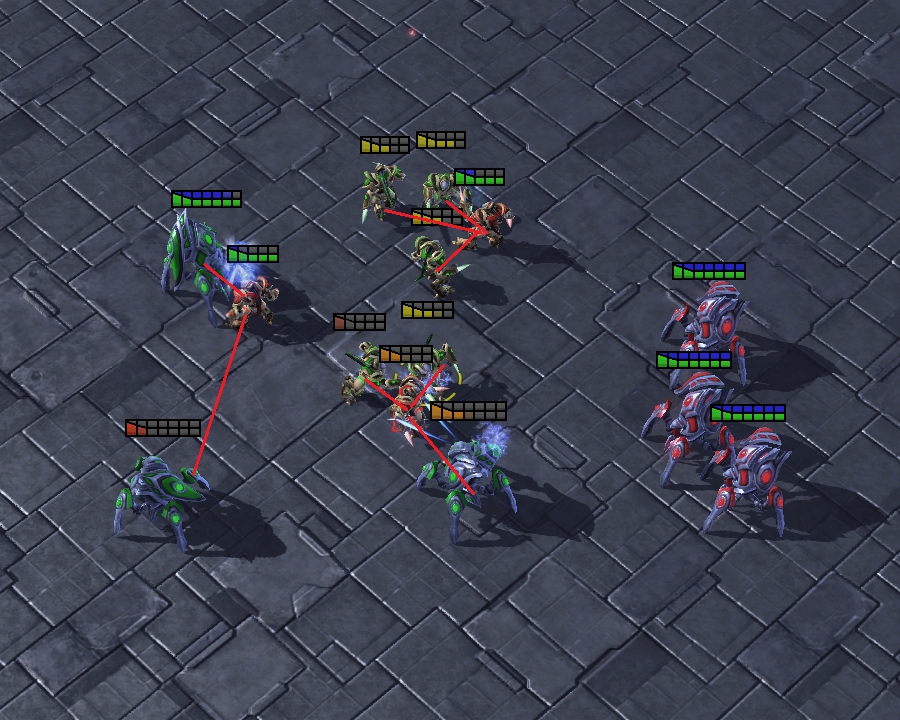}\label{fig-counter}
\end{minipage}
}
\caption{Strategies learned by the agents using UNMAS. There are two basic strategies \emph{more to fight less} and \emph{damage sharing}, and two advanced strategies \emph{block enemies} and \emph{defensive counterattack} to address more difficult scenarios.}
\label{fig:strategy}
\end{figure*}
\subsection{Strategy Analysis}
To understand what UNMAS has learned to achieve the performance, we analyze the strategies of agents using UNMAS according to the combat replay in this subsection.

\begin{enumerate}[labelsep = .5em, leftmargin = 0em, itemindent = 2em]
\item[1)] \emph{Homogeneous Scenarios}: $3m$, $8m$ and $5m\_vs\_6m$. In these scenarios, the combat is conducted between Marines. Next is the analysis of the strategies executed by the multi-agent system.

\begin{enumerate}[leftmargin = 2em]
\item[a)] \emph{More to Fight Less}: The multi-agent system adjusts the formation to form a situation where more agents attack fewer enemies in the combat. It is a basic strategy for micro-management scenario. Taking an example, as shown in Fig. \ref{fig:strategy}(a), agents try to concentrate on attacking an enemy Marine. The red line indicates that the agent is attacking. By adopting this strategy, the multi-agent system is able to eliminate the enemies as quickly as possible to reduce the damage caused by them.

\item[b)] \emph{Damage Sharing}: The agents with higher health share the damage for agents with lower health. Since only alive agents can cause damage to enemies, it is important to ensure the survival of agents who is in danger. As shown in Fig. \ref{fig:strategy}(b), when the health of the agent is low, it retreats a distance from the enemies to allow the other agents to share the damage. The white line indicates that the agent is moving. At the same time, the agents around it step forward in the enemy's direction to ensure successful damage sharing. By adopting this strategy, the agents are able to survive longer to maximize the damage they cause.

\end{enumerate}
\end{enumerate}

In the homogeneous scenarios, the multi-agent system performs the above strategies to defeat the enemies. Although other methods like QMIX can also learn these strategies, the agents using UNMAS make fewer mistakes.

\begin{enumerate}[labelsep = .5em, leftmargin = 0em, itemindent = 2em]
\item[2)] \emph{Heterogeneous Scenarios}: $2s3z$, $3s5z$ and the most difficult map $3s5z\_vs\_3s6z$. In these scenarios, the combat is conducted between two types of units: Zealot and Stalker. The former is short-range and the latter is long-range. For scenarios where units have different roles, the multi-agent system not only needs to execute the strategies analyzed above but also needs new strategies to accommodate the change. Next is the analysis of the strategies executed by the multi-agent system.

\begin{enumerate}[leftmargin = 2em]
\item[a)] \emph{Block Enemies}: The Zealots try to block the enemy Zealots to protect the Stalkers. Since the Stalker is a long-range unit, it is able to cause damage to the enemies without being attacked. As shown in Fig. \ref{fig:strategy}(c), the multi-agent system using UNMAS learns the strategy of placing the Zealots between Stalkers and enemies to ensure the safety of Stalkers. The yellow line is the defensive line formed by Zealots. Through this strategy, the multi-agent system can maximize its damage as much as possible.

\item[b)] \emph{Defensive Counterattack}: Agents focus on eliminating enemy Zealots who break through the defensive line as a defense, and then continue their previous attack as a counterattack. In the actual combat, the defensive line may not completely separate the enemies from Stalkers. As shown in Fig. \ref{fig:strategy}(d), the multi-agent system using UNMAS learns to attack the enemy who breaks through the defensive line. They also perform a \emph{more to fight less} strategy to eliminate enemies as quickly as possible. In addition, sometimes one Stalker chooses to attack the enemies at the edge of the battlefield to disperse pressure from other allies. 

\end{enumerate}
\end{enumerate}

It should be noted that, in the heterogeneous scenarios, the multi-agent system not only learns the strategies mentioned above, but also learns \emph{more to fight less} and \emph{damage sharing} as in homogeneous scenarios. Other methods all fail in $3s5z\_vs\_3s6z$ because they just perform the \emph{more to fight less} strategy as in homogeneous scenarios, while UNMAS achieves 28\% winning rate because of the above strategies.

\section{Conclusion}
This paper proposes UNMAS, a novel multi-agent reinforcement learning method that is more adaptable to the unshaped scenario, in which the number of agents and the size of action set change over time. UNMAS factorizes the joint action-value by mapping Q values nonlinearly and calculating its weights in the joint action-value for each agent. We theoretically analyze that the value factorization of UNMAS meets the IGM condition. In order to adapt to the change in the size of action set, the individual action-value network uses two network streams to evaluate the actions in \emph{environment-oriented} subset and \emph{unit-oriented} subset.

We compare UNMAS with VDN, QMIX, QTRAN, and ASN in StarCraft II micro-management scenario. Our results show that UNMAS achieves state-of-the-art performance among tested algorithms and also learns effective strategies in the most difficult scenario $3s5z\_vs\_3s6z$ that other algorithms fail in.

\section*{Acknowledgment}
The authors would like to thank Professor Hongbin Sun of Xi'an Jiaotong University for his suggestions. 


%

\appendices
\section{Environmental Settings}
\setcounter{table}{0}
\renewcommand{\thetable}{A.\arabic{table}}
We use SMAC as the experimental environment. Alive agents obtain their local observations from environment and execute their actions. The observation provided by SMAC consists of the information about agent self, other agents, and the enemies within sight range. In detail, the information contains the following elements:
\begin{enumerate}
\item \emph{Distance}: the distance between agent and other units;
\item \emph{Relative coordinates}: the relative coordinates of x and y between agent and other units;
\item \emph{Health}: the health percentage of agent and other units;
\item \emph{Shield}: the shield percentage of agent and other units;
\item \emph{Unit type}: the one-hot coding for the unit type of agent and other units;
\item \emph{Last action}: the action executed by agent at last timestep;
\item \emph{Agent Index}: the index used to distinguish agents.
\end{enumerate}

Similarly, the global state also provides the above information. However, the difference is that the relative information takes the center point of the map as the reference point. The information included in the state contains all the alive agents on the map instead of those only within the \emph{sight range}. 

\section{Parameter Settings}
\setcounter{table}{0}
\renewcommand{\thetable}{B.\arabic{table}}
In the experiments, all methods adopt the \emph{same} hyper parameters, which are shown in Table \ref{table:hp}, and are the default values in PyMARL. 

\begin{table}[ht]
\caption{Hyper-parameters of experimental methods, including ASN, VDN, QMIX, QTRAN, RODE, and UNMAS.}
\label{table:hp}
\renewcommand{\arraystretch}{1.2}
\begin{center}
\begin{tabular}{|m{2cm}<\centering|m{2.75cm}<\centering|m{2.25cm}<\centering|}
\hline
Setting			  		 				& Name         						& Value  				\\
\hline
\multirow{7}{*}{Training Settings}   	 			& Size of Replay buffer \emph{D}		& 5000 episodes			\\ \cline{2-3}
   						 				& Batch size $b$					& 32 episodes			\\ \cline{2-3}
    										& Testing interval					& 10000 timesteps  		\\ \cline{2-3}
    										& Target update interval				& 200 episodes			\\ \cline{2-3}
    										& Maximum timesteps					& 2 million timesteps		\\ \cline{2-3}
    										& Exploration rate $\epsilon$			& 1.0 to 0.05			\\ \cline{2-3}
    					 					& Discount factor $\gamma$				& 0.99 				\\
\hline
\multirow{7}{*}{Network Settings}		 		& Self-weighting mixing network unit		& 32 					\\ \cline{2-3}
	   					 				& Hyper network unit					& 64 					\\ \cline{2-3}
    					 					& GRU layer unit					& 64 					\\ \cline{2-3}
    					 					& Optimizer						& RMSProp				\\ \cline{2-3}
    					 					& RMSProp $\alpha_R$					& 0.99 				\\ \cline{2-3}
    					 					& RMSProp $\epsilon_R$				& 0.00001				\\ \cline{2-3}
    					 					& Learning rate $\alpha$				& 0.0005				\\
\hline
\end{tabular}
\end{center}
\end{table}

The hypernetwork of UNMAS, which is used to calculate the weights and biases of self-weighting mixing network, adopts the same settings as QMIX. $v_{t}$, the final element of self-weighting mixing network, is the output of a network with two layers and one ReLU activation. Other parameters related to the networks are also shown in Table \ref{table:hp}.

\bibliographystyle{IEEEtran}
\bibliography{IEEEabrv, ref}

\begin{IEEEbiography}[{\includegraphics[width=1in,height=1.25in,clip,keepaspectratio]{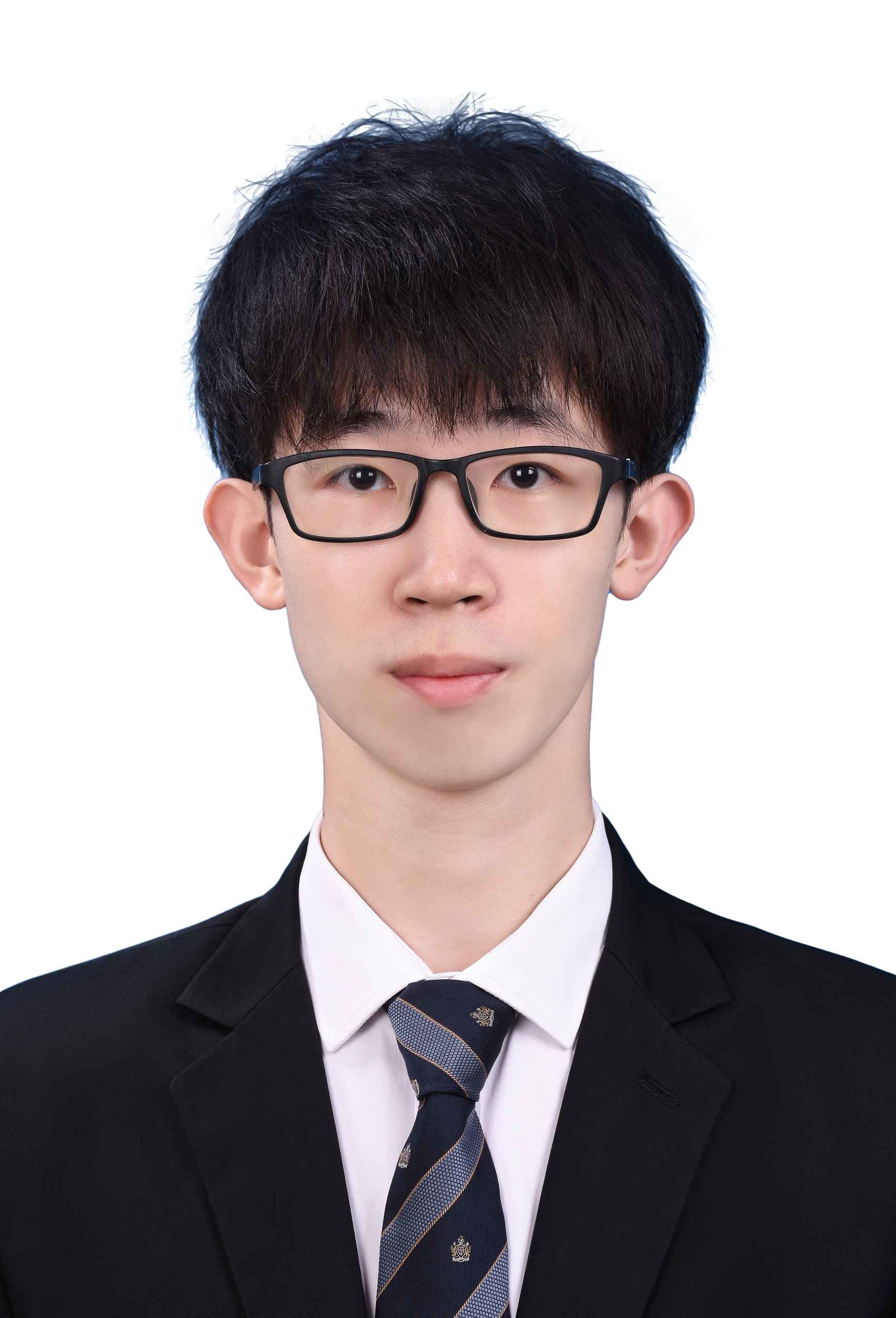}}]{Jiajun Chai} 
received the B.S degree from the Faculty of Electronic and Information Engineering, Xi’an Jiaotong University, Xi’an, China, in 2020. He is currently pursuing the Ph.D. degree with the State Key Laboratory of Management and Control for Complex Systems, Institute of Automation, Chinese Academy of Sciences, Beijing, China. 

His current research interests include multi-agent reinforcement learning, deep learning, and game AI. 
\end{IEEEbiography}

\begin{IEEEbiography}[{\includegraphics[width=1in,height=1.25in,clip,keepaspectratio]{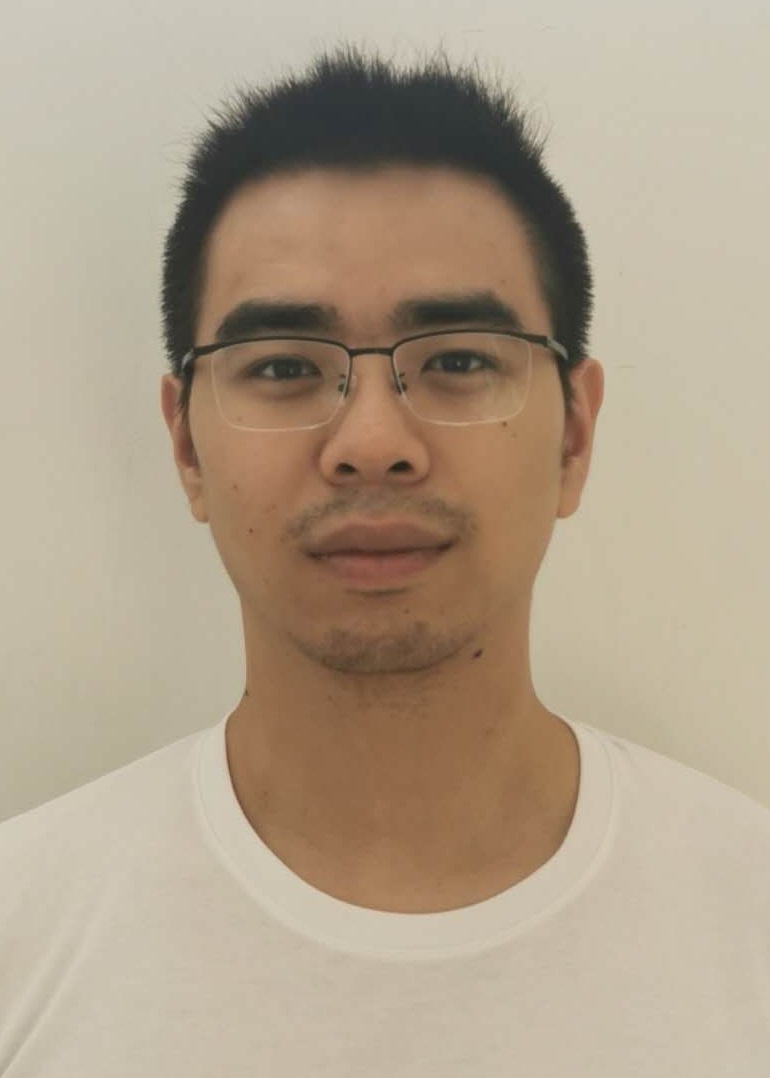}}]{Weifan Li} 
received the M.E degree in the Mechanical engineering and automation, Fuzhou University, Fuzhou, China, in 2018. He is currently pursuing the Ph.D. degree in control theory and control engineering at the State Key Laboratory of Management and Control for Complex Systems, Institute of Automation, Chinese Academy of Sciences, Beijing, China. 

His current research interest is deep reinforcement learning.
\end{IEEEbiography}

\begin{IEEEbiography}[{\includegraphics[width=1in,height=1.25in,clip,keepaspectratio]{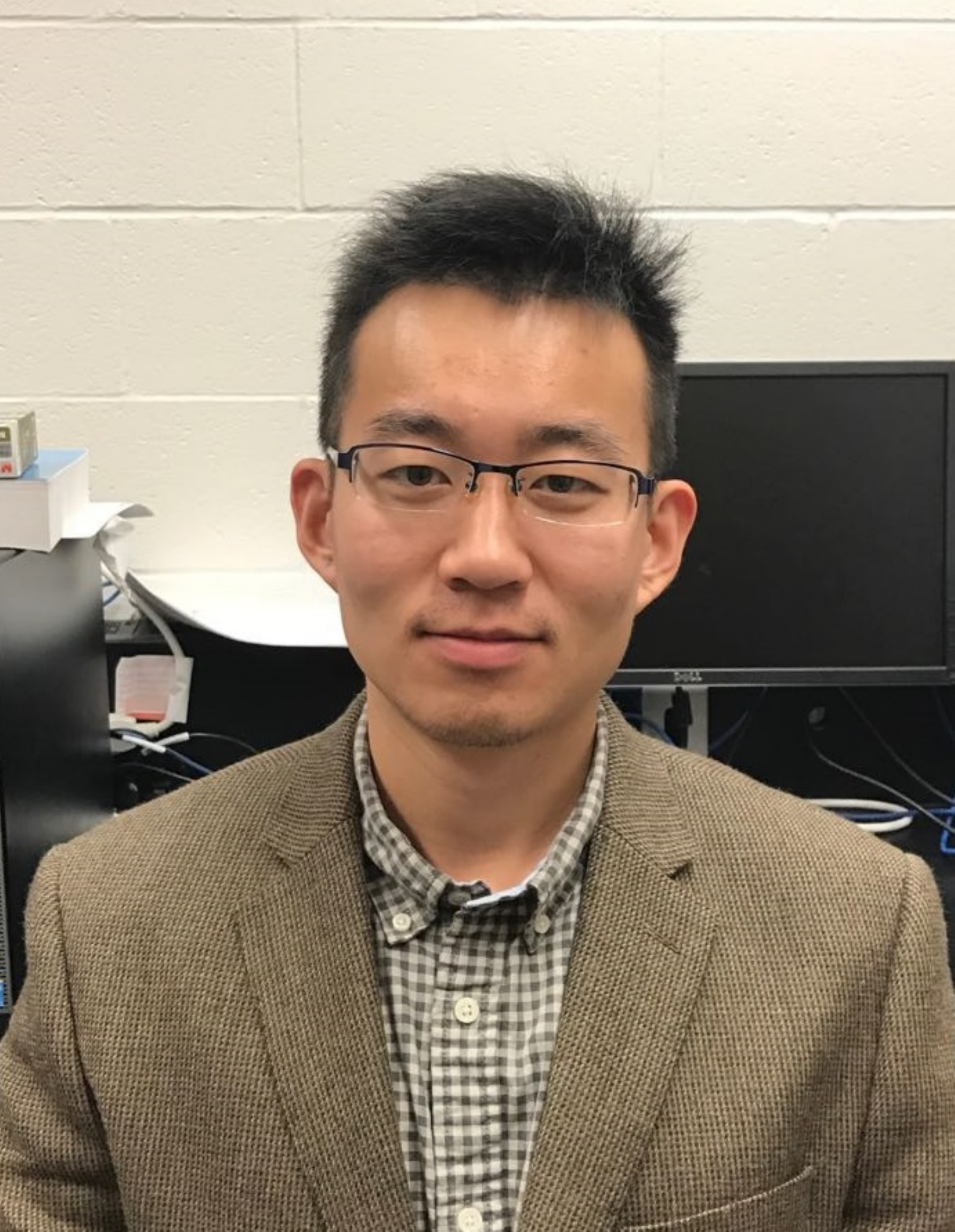}}]{Yuanheng Zhu} 
(M'15) received the B.S. degree in automation from Nanjing University, Nanjing, China, in 2010, and the Ph.D. degree in control theory and control engineering from the Institute of Automation, Chinese Academy of Sciences, Beijing, China, in 2015. He is currently an Associate Professor with the State Key Laboratory of Management and Control for Complex Systems, Institute of Automation, Chinese Academy of Sciences. His research interests include optimal control, adaptive dynamic programming, reinforcement learning, automatic driving, and game intelligence.
\end{IEEEbiography}

\begin{IEEEbiography}[{\includegraphics[width=1in,height=1.25in,clip,keepaspectratio]{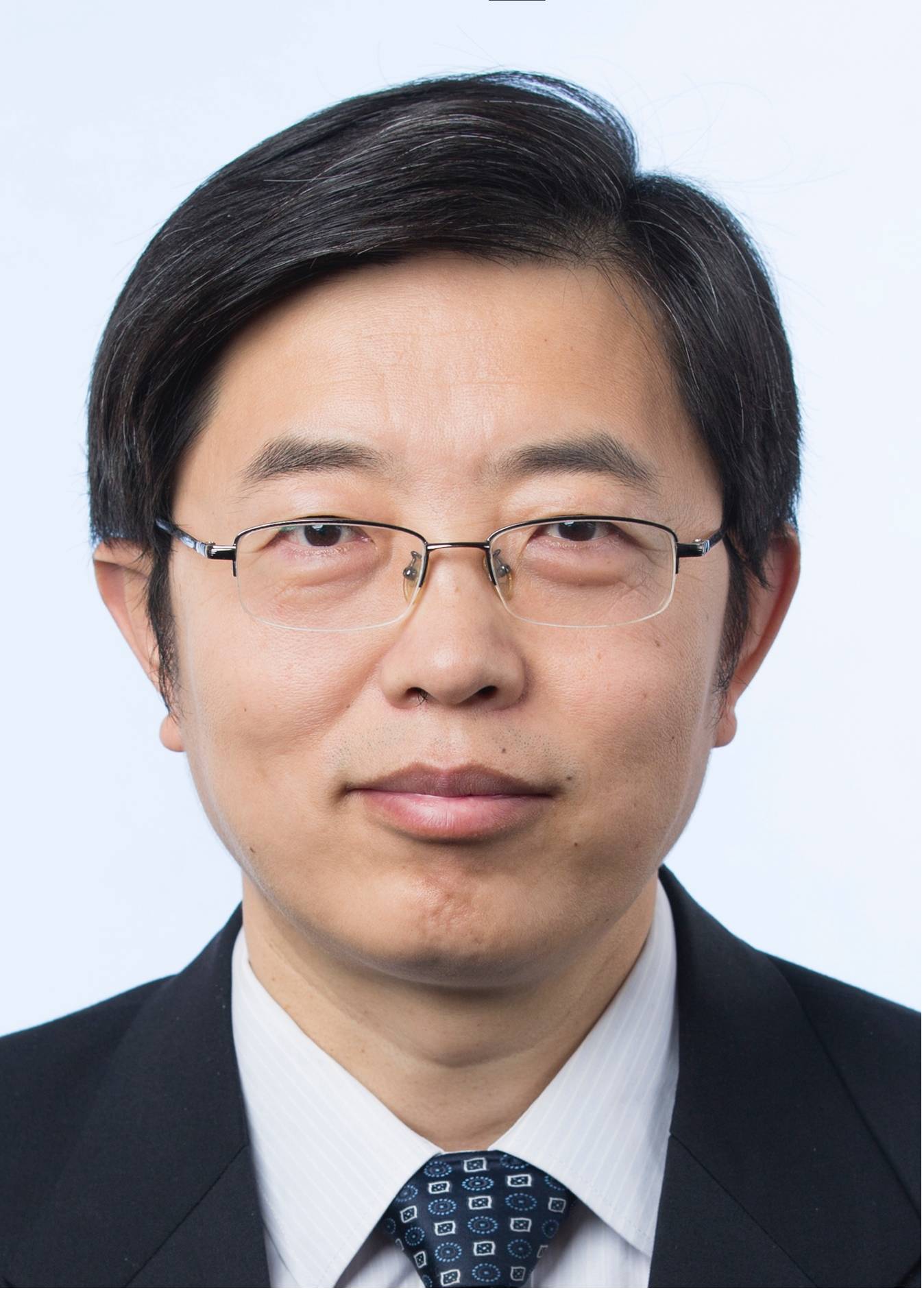}}]{Dongbin Zhao} 
(M’06-SM’10-F’20) received the B.S., M.S., Ph.D. degrees from Harbin Institute of Technology, Harbin, China, in 1994, 1996, and 2000 respectively. He is now a professor with Institute of Automation, Chinese Academy of Sciences, and also with the University of Chinese Academy of Sciences, China. He has published 6 books, and over 100 international journal papers. His current research interests are in the area of deep reinforcement learning, computational intelligence, autonomous driving, game artificial intelligence, robotics, etc.

Dr. Zhao serves as the Associate Editor of IEEE Transactions on Neural Networks and Learning Systems, IEEE Transactions on Cybernetics, IEEE Transactions on Artificial Intelligence, etc. He is the chair of Distinguished Lecture Program of IEEE Computational Intelligence Society (CIS). He is involved in organizing many international conferences. He is an IEEE Fellow.
\end{IEEEbiography}

\begin{IEEEbiography}[{\includegraphics[width=1in,height=1.25in,clip,keepaspectratio]{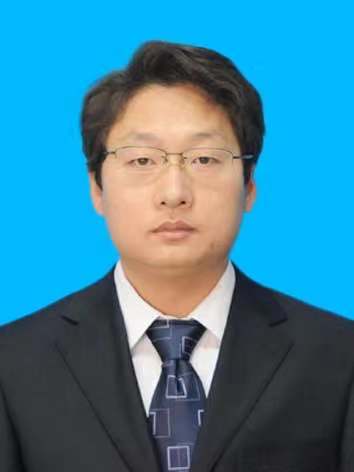}}]{Zhe Ma} 
has received his Ph.D. degree. He is currently a researcher at X-Lab in the Second Academy of CASIC. 

His current research interests include artificial intelligence and SoS. 
\end{IEEEbiography}

\begin{IEEEbiography}[{\includegraphics[width=1in,height=1.25in,clip,keepaspectratio]{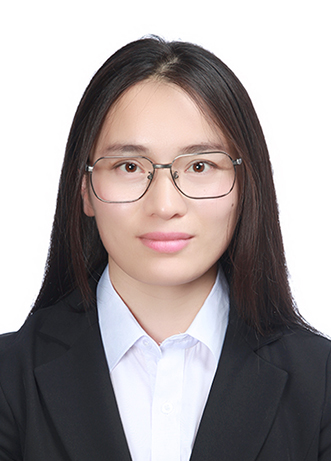}}]{Kewu Sun} 
has received her B.S. and M.S. degree. She is currently a senior engineer at X-Lab in the Second Academy of CASIC. 

Her current research interests include multi-agent reinforcement learning and SoS.
\end{IEEEbiography}

\begin{IEEEbiography}[{\includegraphics[width=1in,height=1.25in,clip,keepaspectratio]{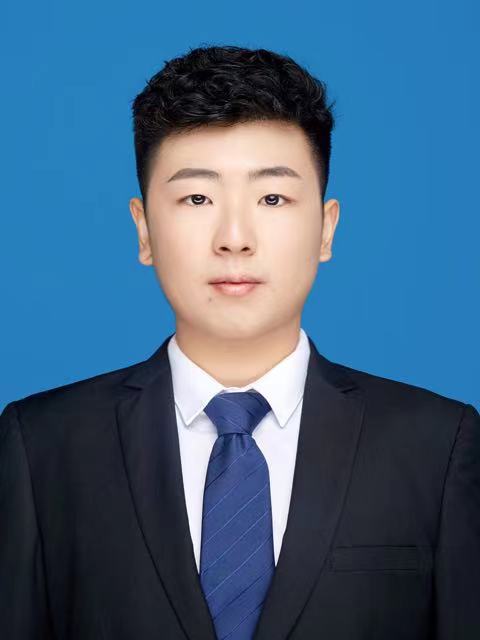}}]{Jishiyu Ding} 
received the B.S. degree from Beijing Jiaotong University in 2015 and Ph.D. degree from Tsinghua University in 2020. He is currently an engineer at X-Lab in The Second Academy of CASIC. 

His current research interest is multi-agent reinforcement learning.
\end{IEEEbiography}

\end{document}